\documentclass[11pt]{article}
\usepackage[affil-it]{authblk}

\usepackage{color}
\usepackage{amsmath,amsfonts,amssymb,mathrsfs,amsthm}
\usepackage{graphicx}
\usepackage{listings}
\usepackage{enumerate,enumitem}
\usepackage{mathtools}
\usepackage[a4paper, total={6in, 8in}]{geometry}
\usepackage{cite}
\usepackage[colorlinks,linkcolor=blue,citecolor=blue]{hyperref}
\usepackage{empheq}
\usepackage{pdfpages}
\usepackage[titletoc]{appendix}
\usepackage[utf8]{inputenc}
\usepackage{xspace}
\numberwithin{equation}{section}

\newcommand{\Eqref}[1]{Eq.~\eqref{#1}}
\newcommand{\Eqsref}[1]{Eqs.~\eqref{#1}}

\newcommand{\Sectionref}[1]{Section~\ref{#1}}
\newcommand{\subfig}[2]{ Fig.~\hyperref[#1]{\ref{#1}#2}}

\newcommand{\T}{\mathcal{T}}

\newtheorem{Thm}{Theorem}
\newtheorem{Cor}{Corollary}
\newcommand{\Defref}[1]{Definition~\ref{#1}}
\newtheorem{defn}{Definition}

\newcounter{mnotecount}[section]
\let\oldmarginpar\marginpar
\renewcommand\marginpar[1]{\-\oldmarginpar[\raggedleft\footnotesize #1]%
	{\raggedright\footnotesize #1}}

\title{A general formalism for coupling scalar fields to the Einstein equations without a variational principle}
\author{J. Ritchie\footnote{Email: josh.ritchie@otago.ac.nz}}
\affil{Department of Mathematics and Statistics, University of Otago, New Zealand.}

\begin{document}
	\maketitle
	\begin{abstract}
		The purpose of this work is to discuss how matter fields are coupled to gravity within the framework of General Relativity. Our particular focus here is on the coupling of scalar field models. In a first step, we suggest a new method for coupling scalar fields to the Einstein equations \emph{without} the use of a variational principle or Lagrangian. We show that, under the appropriate assumptions, this new method (for coupling scalar fields to gravity) reproduces the minimally and $k$-essence scalar field couplings with a non-zero potential. We therefore interpret this formalism as describing a \emph{generic} method for coupling scalar fields to gravity. The approach described here allows for a number of free fields which we interpret as constitutive freedoms. In a second step, we choose these free fields in such a way that the resulting system is somehow ``near minimal''. In this setting we investigate Bianchi I type solutions. We establish conditions under which the solutions are asymptotically Kasner, near the initial singularity, and investigate their stability properties.    
	\end{abstract}
	\section{Introduction}
	Matter fields play a key role within the framework of General Relativity (GR), providing physical sources that cause spacetime to curve thereby generating a gravitational field \cite{Gravitation}. Indeed, there are many types of matter fields employed in GR, that allow researchers to study a diverse range of matter-driven phenomena. For example, perfect fluids are often employed to model astrophysical and cosmological matter\cite{Wilson:1945}, including stars \cite{Bayin:1982,Pfister:2011,Liebling:2017}, accretion disks \cite{Korobkin:2011,Kim:2018}, large-scale cosmic evolution \cite{Andersson:2021,Tavakoli:2019,WainwrightEllis1997,BeyerMarshallOliynyk:2023}, and stable Big Bang formation \cite{BeyerMarshallOliynyk:2024,BeyerOliynyk:2023_2,BeyerOliynyk:2024,BeyerOliynyk:2026}. In addition, \emph{imperfect} fluids can be used to account for thermodynamic effects such as viscosity and heat conduction \cite{Israel:1979,Maartens:1995,Maartens:1996}. Conversely, Electromagnetic fields --governed by the Maxwell equations-- are central to the study of charged black holes \cite{deOliveira:1994,McMaken:2023}, magnetized plasmas \cite{giacomazzo:2012}, and charged neutron stars \cite{Ghezzi:2005,Rossetto:2023}.  
	
	Of particular interest for the present work are \emph{scalar field} models. This type of matter occupies a particularly prominent position in the standard model cosmology largely due to its central role in \emph{inflation}. In the standard model of inflationary cosmology, a scalar field --commonly referred to as the ``inflaton''-- is introduced to provide a dynamical mechanism responsible for generating a period of accelerated expansion in the early universe commonly referred to as ``inflation'' \cite{Guth:1981,Lazarides:2001}. Inflation ends with reheating, during which the inflatons energy is transferred into field excitations that thermalize into a relativistic plasma and drive radiation-dominated expansion \cite{Amin:2015}. During this period of thermalization, the scalar field plays a central role in generating the density perturbations that seed the formation of large-scale structure in the early universe \cite{Sriramkumar:2009,Sasaki:2012,Olle:2020,Cyncynates:2021,Mahbub:2023}. Moreover, scalar fields are often studied in the context of stable Big Bang formation \cite{RodnianskiSpeck:Linear,RodnianskiSpeck:NonLinear,LapseScalarField,Ritchie2022}. This is largely due to the fact that a canonical scalar field with a time-like gradient behaves like a stiff perfect fluid (that is, a fluid whose speed-of-sound is equal to the speed-of-light) \cite{Isenberg1990}. This correspondence means that, near the Big Bang, scalar fields can suppress the so called mixmaster oscillations, that are typically present when a stiff fluid is not included \cite{MixMasterUniverse,OscillatoryApproach:BKL,Beguin2010,BeguinDutilleul2023}. 
	
	There are several different types of scalar field models that have been studied throughout the literature. These include (but are not limited to) minimally-, conformally-, and $k$-essence-coupled scalar fields \cite{Guth:1981,Faraoni:2013,Armendariz-Picon:2001}. Arguably the most important of these is the minimally coupled scalar field which was initially introduced as a solution of the so called \emph{horizon problem} \cite{Guth:1981}. We refer the interested reader to \cite{Liddle:1999,Sloan:2007} for a detailed discussion. Similarly, non-canonical $k$-essence scalar field models have been introduced as a possible solution to the \emph{cosmic coincidence problem} \cite{Armendariz-Picon:2001}.
	
	These models have been undeniably successful and their importance should not be understated. However, that does not mean these approaches do not come without their own challenges. Early minimally coupled scalar field models --which are of particular interest for our work here-- struggled with the \emph{Graceful Exit problem}, where the universe was unable to properly transition out of inflation. This problem was resolved by the introduction of the \emph{slow roll conditions} \cite{Linde:1982,Guth:1984,Beyer:2013}. 
	
	Despite the success of single-field slow-roll models, several theoretical and observational challenges have led researchers to introduce models involving multiple interacting scalar fields \cite{Kaiser2016,Bertolami:2012,GarciaArroyo:2024}. From a particle physics perspective, high-energy frameworks such as String Theory or Supergravity naturally predict a ``landscape'' populated by numerous scalar fields, often referred to as moduli or dilatons \cite{Wands:2007,berglund2009multi}. In these contexts, it is arguably more natural to consider inflation as a collective process driven by multiple fields rather than a single, isolated inflaton. Furthermore, multi-field setups like \emph{hybrid inflation} offer a more sophisticated resolution to the end of inflation; here, one field drives the expansion while a second ``waterfall'' field triggers a rapid phase transition, ensuring a complete reheating process that avoids the lingering issues of simpler models \cite{gong2011waterfall}.
	
	Moreover, multi-field models resolve the rigid observational constraints of the single-field paradigm. While single-field inflation typically predicts purely adiabatic perturbations \cite{Weinberg:2008}, multi-field interactions allow for the existence of \emph{isocurvature} (or entropy) perturbations \cite{hwang2002cosmological}. These interactions also provide a mechanism for generating detectable non-Gaussian statistical signatures in the Cosmic Microwave Background \cite{Wands:2007}.
	
	Crucially, all of these approaches rely on a Lagrangian formulation, from which both the equations of motion and the energy-momentum tensor are derived. While this framework is powerful, it excludes certain types of non-linear phenomena that are mathematically characteristic of non-variational systems. These include, but are not limited too, certain classes of instabilities and feedback mechanisms that cannot be derived from a standard action principle. This limitation raises the following question: Is it possible to couple a more general class of scalar field models to gravity without the use of a Lagrangian formulation? We are not the first to ask this question. Recent work by Gao et. al. has demonstrated that certain non-Lagrangian scalar field couplings can indeed be consistently embedded within the framework of GR \cite{Gao:2010}. In \cite{Gao:2010}, Gao et. al. showed that it is possible to consistently couple scalar fields to gravity without starting from a Lagrangian and constructed an energy-momentum tensor directly from the field, rather than from a variational principle. However, their approach has some notable limitations. First, they mostly focused on single-field models with a particular type of derivative coupling; meaning that many more general forms of scalar dynamics are not included. Second, while their formalism provides a specific method for generating an energy-momentum tensor, it does not fully characterize the space of allowed couplings. i.e., it does not systematically identify all the freedoms available when constructing non-Lagrangian scalar fields in GR. These limitations leave open the question of whether a truly general framework for non-Lagrangian scalar fields can be formulated. The goal of the present work is to develop one such framework in such a way as to ensure that, under the appropriate reductions, known scalar field couplings (such as the minimal and $k$-essence couplings) are recovered. 
	
	To do this, we ``split'' the energy momentum tensor into two pieces which we refer to as the \emph{kinetic} and \emph{potential} parts. In this framework, the scalar fields equation of motion and the potential part (of the energy momentum tensor) are treated as freely specifiable, while the kinetic part is determined as a solution of the Bianchi identity. This divergence-free condition (of the energy-momentum tensor) is under-determined and as such it is necessary to choose some portion of the kinetic part of the energy momentum tensor. A large portion of our work here is dedicated to discussing exactly this issue. In particular, by making use of the $(3+1)$-decomposition we express the Einstein+matter equations as a Cauchy problem and establish under what conditions our approach here leads to a well-posed initial value problem.  
	
	Our approach here is of course mathematically valid. However, its \emph{physical relevance} is unclear. Indeed, there are arguments in the literature that suggest, as a matter of principle, that scalar fields ``should'' arise from a Lagrangian formulation (see, for example, \cite{Durrer:2008}). This is a valid critique of our framework. However, it is worth noting that the standard cosmological model relies heavily on the perfect fluid model, which itself is often treated phenomenologically rather than being derived from a fundamental action. Our approach here extends this established phenomenological liberty to the scalar field sector, treating it as an effective medium governed by constitutive relations rather than a strictly variational entity. Moreover, the physical relevance of such an approach cannot be established \emph{before} it has been investigated.
	
	This paper is organised as follows. In \Sectionref{Sec:A_generic_framework_for_coupling_scalar_fields_to_gravity} we introduce our framework for coupling matter to gravity. In a first step, in \Sectionref{Sec:A_method_for_coupling_scalar_fields_to_gravity}, we discuss our approach for any type of matter. We then focus in on a single scalar field whose dynamics is governed by a wave-type equation. Our focus on a single scalar field here simplifies much of the discussion. However, it should be emphasised that our approach can easily be generalised to include multiple interacting matter fields. In \Sectionref{Sec:The_(3+1)-decomposition_and_evolution_equations} we then discuss the $(3+1)$-decomposition and establish under what conditions the resulting system produces a well-posed Cauchy problem. As mentioned above, our framework allows for a number of ``free data''. In \Sectionref{Sec:Bianchi_I} we make specific choices of said free data and investigate the corresponding class of Bianchi I solutions. In particular, we demonstrate that, under appropriate conditions, these particular choices permit the existence of Kasner-type solutions. In addition, we show that these solutions are non-linearly stable under perturbation.

	\section{A generic framework for coupling scalar fields to gravity}
	\label{Sec:A_generic_framework_for_coupling_scalar_fields_to_gravity}
	In this section we describe our framework for coupling scalar fields to the Einstein equations. In \Sectionref{Sec:A_method_for_coupling_scalar_fields_to_gravity} we describe our generic framework which essentially splits the energy-momentum tensor into kinetic and potential pieces. We then provide a number of examples demonstrating that this approach is consistent with standard scalar field models. In \Sectionref{Sec:The_(3+1)-decomposition_and_evolution_equations} we discuss the $(3+1)$-decompositions and establish under what conditions our matter model produces a well-posed initial value problem. Finally, in \Sectionref{Sec:Reductions} we show that minimal and $k$-essence scalar field models (with a non-zero potential) can be recovered within our formalism.  
	
	\subsection{A method for coupling scalar fields to gravity}
	\label{Sec:A_method_for_coupling_scalar_fields_to_gravity}
	In this section here we discuss the Einstein equations and describe our matter coupling procedure. To that end, we consider a globally hyperbolic, $4$-dimensional\footnote{Although we restrict our attention here to $4$-dimensional spacetimes, our approach can also work for $n$-dimensional spacetimes.}, smooth Lorentzian manifold $(M, g_{\alpha\beta})$ where $g_{\alpha\beta}$ is a smooth Lorentzian metric (i.e., a metric with signature $(-,+,+,+)$). Note here that we use Greek letters for indices that run from $0,\dots, 3$. The Einstein equations --in geometric units ($c=8\pi G=1$ for the speed of light $c$ and the gravitational constant $G$)-- can be written as\footnote{Here, the prefix ${}^{(4)}$ refers the number of dimensions.}, 
	\begin{align}
		\prescript{(4)}{}{R}_{\mu\nu}  = T_{\mu\nu}-\frac{1}{2}Tg_{\mu\nu},
		\label{Eq:EFEs}
	\end{align}
	where $\prescript{(4)}{}{R}_{\mu\nu},\prescript{(4)}{}{R}$ are the Ricci tensor and scalar (associated with $g_{\mu\nu}$), respectively, $T_{\mu\nu}$ is the energy momentum tensor of the matter fields with $T=g^{\mu\nu}T_{\mu\nu}$ and where $\nabla_{\mu}$ is the unique Levi-Civita connection associated with $g_{\mu\nu}$. In regards to the matter sector, let $u^I$ be a collection of $N$-unknowns with $I=1,\dots,N$ that constitute the matter content of our spacetime\footnote{It should be understood that we have not yet restricted ourselves to scalar fields models. The collection $u^I$ can be scalar fields and/or the components of a vector field or some other tensor field.} $(M, g_{\alpha\beta})$. Given $u^I$ we pick some symmetric tensorial operator $V_{\mu\nu}=V_{\mu\nu}[u^I,g_{\mu\nu}]$ whose components are determined in terms of $u^I$ and its derivatives $\nabla_{\mu}u^I$, with the property that ${\nabla_{\mu}{V}^\mu}_\nu\ne 0$. Given such a tensor, we write $T_{\mu\nu}$ as 
	\begin{subequations}
		\begin{align}
			T_{\mu\nu}=\T_{\mu\nu}-V_{\mu\nu},
			\label{Eq:GeneralT}
		\end{align}
		where $\T_{\mu\nu}$ is a symmetric tensor ($\T_{\mu\nu}=\T_{\nu\mu}$) whose components are determined from the requirement that $T_{\mu\nu}$ is divergence-free. i.e., 
		\begin{align}
			\nabla^{\mu}\T_{\mu\nu} = \nabla^{\mu}V_{\mu\nu}. 
			\label{Eq:Div_Tdd}
		\end{align}
	\end{subequations}
	One may think of the tensor $V_{\mu\nu}[u^I,g_{\mu\nu}]$ as somehow representing the \emph{potential} coupling. In the case of a minimally coupled scalar field (so that $u^I=\{\phi\}$), it is exactly the potential $V(\phi)g_{\mu\nu}$. Conversely, $\T_{\mu\nu}$ represents the \emph{dynamic} or \emph{kinetic} part of the energy momentum tensor. From this perspective, one thinks of \Eqref{Eq:Div_Tdd} as a transfer equation between the potential and kinetic sectors; the non-vanishing divergence of $V_{\mu\nu}$ acts as a source term for the kinetic sector $\T_{\mu\nu}$, ensuring that the total energy-momentum remains divergence-free as required by the Bianchi identities. Observe carefully, however, that \Eqref{Eq:Div_Tdd} is a set of $4$-equations, which are intended to determine the $10$-components of $\T_{\mu\nu}$. It follows then that the system is under-determined and hence some part of $\T_{\mu\nu}$ must be specified \emph{before} \Eqref{Eq:Div_Tdd} can be solved. This leads us to the following two questions: (1) How should the ``potential part'' $V_{\mu\nu}[u^I,g_{\mu\nu}]$ be chosen? (2) What part of $\T_{\mu\nu}$ should be considered as freely specifiable? To address these questions we must first make a specific choice of matter field(s). So, suppose that there is a real-valued scalar field $\phi:M\rightarrow\mathbb{R}$ which constitutes the matter content of the universe (i.e., $u^I=\{\phi\}$). The equation of motion for $\phi$ is a freely specifiable. Indeed, there are many ways in which one could choose $\phi$ leading to a diverse range of dynamics. However, for the sake of simplicity, in this paper here we shall restrict our attention to scalar fields that satisfy a wave equation of the form
	\begin{align}
		B^{\mu\nu}\nabla_{\mu}\nabla_{\nu}\phi + \tilde{B}^{\mu}\nabla_{\mu}\phi = f(\phi,\mathcal{X}),
		\quad
		\mathcal{X} = -\frac{1}{2}g^{\mu\nu}\nabla_{\mu}\phi\nabla_{\nu}\phi,
		\label{Eq:EoM}
	\end{align}
	where $B^{\mu\nu}$ is a Lorentz type metric, $\tilde{B}^{\mu}$ is an arbitrary $4$-vector, and $f(\phi,\mathcal{X})$ is a freely specifiable function. Note that $B^{\mu\nu}$ and $\tilde{B}^\mu$ are allowed to depend on $\phi$ and its first derivatives $\nabla_\mu\phi$. Scalar field couplings of this type (in which $\phi$ is given by \Eqref{Eq:EoM} and $T_{\mu\nu}$ by \Eqref{Eq:GeneralT}) appear commonly throughout the literature, although it is not typically discussed in this way. To justify this statement, we now consider three specific examples. 
	
	\paragraph{Example: Minimally coupled scalar field.} For our first example we show that a minimally coupled scalar field is consistent with our approach here. This model is perhaps the most common way to couple scalar fields to gravity and, as was noted above, plays a significant role the standard model of cosmology. This type of matter is typically derived from a Lagrangian and is therefore conservative (variational). Within our framework, the appropriate equation of motion can be obtained by setting
	\begin{subequations}
		\label{Eq:Minimial}
		\begin{align}
			B^{\mu\nu} = g^{\mu\nu},
			\quad
			\tilde{B}^\mu = 0,
			\quad
			V_{\mu\nu}=V(\phi)g_{\mu\nu},
			\quad
			f(\phi,\mathcal{X})=V^\prime(\phi).	
			\label{Eq:MinimalFreeData}
		\end{align} 
		Then, in this case, one readily checks that 
		\begin{align}
			\T_{\mu\nu}=\nabla_{\mu}\phi\nabla_{\nu}\phi -\frac{1}{2}\nabla_{\sigma}\phi\nabla^{\sigma}\phi\, g_{\mu\nu},
			\label{Eq:MinimallyCoupled}
		\end{align}
		is a \emph{particular} solution of \Eqref{Eq:Div_Tdd}. This energy momentum tensor (along with the particular equation of motion \Eqref{Eq:MinimalFreeData}) describes the standard minimally coupled scalar field. 
	\end{subequations}
	
	\paragraph{Example: $k$-essence scalar field.} For our second example, we show that $k$-essence scalar field models are consistent with our approach here. This particular type of scalar field model --which is typically derived from a Lagrangian-- was initially introduced to solve the cosmic coincidence problem (see, for example, \cite{velten2014aspects} for an explanation of cosmic coincidence and \cite{Vikman:2007} for a discussion of $k$-essence). Within our framework, the appropriate equation of motion is obtained by setting
	\begin{subequations}
		\begin{align}
			B^{\mu\nu} = -\frac{\partial V}{\partial\mathcal{X}} g^{\mu\nu},
			\quad
			\tilde{B}^\mu = -\nabla^{\mu}\left( \frac{\partial V}{\partial\mathcal{X}} \right),
			\quad
			V_{\mu\nu}=V\, g_{\mu\nu}
			\quad
			f(\phi,\mathcal{X})=\frac{\partial V}{\partial\phi}.	
		\end{align}
		Note that, for the sake of readability, we have suppressed the arguments of $V=V(\phi,\mathcal{X})$. In this case, the equation of motion \Eqref{Eq:EoM} can be written as 
		\begin{align}
			\nabla_{\mu}\left( \frac{\partial V}{\partial\mathcal{X}}\nabla^\mu\phi \right) = \frac{\partial V}{\partial\phi}.
			\label{Eq:Box_q}
		\end{align}
		Then, one readily checks that 
		\begin{align}
			\T_{\mu\nu}=\frac{\partial V}{\partial \mathcal{X}}(\phi,\mathcal{X})\nabla_{\mu}\phi\nabla_{\nu}\phi,
			\label{Eq:kessence}
		\end{align}
		is a \emph{particular} solution of \Eqref{Eq:Div_Tdd}. This energy momentum tensor (along with the particular equation of motion \Eqref{Eq:Box_q}) describes a $k$-essence scalar field. 
	\end{subequations}
	\paragraph{Example: Unified dark-matter/dark-energy scalar field model.} For our third example we demonstrate that Gao's \emph{unified dark-matter/dark-energy scalar field model} can also be obtained within our framework here. Unlike the minimally-coupled case, this model has received comparatively less attention in the literature, despite its notable feature that it is not generated from a Lagrangian formulation. In \cite{Gao:2010}, the relevant energy conservation equations --derived as a consequence of \Eqref{Eq:Div_Tdd}-- were established under the special assumption of spatial homogeneity and isotropy. Nevertheless, their scalar field equation of motion can be recovered within our framework by setting
	\begin{subequations}
		\begin{align}
			B^{\mu\nu} = \frac{1}{2}\left( g^{\mu\nu} - \tilde{u}^\mu\tilde{u}^\nu \right),
			\quad
			\tilde{B}^\mu =0,
			\quad
			f(\phi,\mathcal{X})=-\Lambda^{\prime}(\phi),
		\end{align}
		for some function $\Lambda(\phi)$, and where we have defined
		\begin{align}
			\tilde{u}^\mu = (2\mathcal{X})^{-1/2}\nabla^\mu\phi.
		\end{align}
		Comparing the formalism of \cite{Gao:2010} with ours yields the relation\footnote{It is worth briefly noting here that Gao et. al. claim that the tensor $V_{\mu\nu}$, as defined in \Eqref{Eq:Gao_Vdd}, is divergence free. The validity of this claim is not clear. Certainly it is true that $\nabla^\mu V_{\mu\nu}=0$ implies that their equation of motion must hold. However, the reverse implication is less clear. Nevertheless, if their claim is true, then their formalism is still consistent with our approach provided we set $V_{\mu\nu}=\Lambda(\phi)g_{\mu\nu}$. In this case  $\T_{\mu\nu}=\nabla_{\mu}\phi\nabla_{\nu}\phi/2$ is a particular solution of \Eqref{Eq:Div_Tdd}.}
		\begin{align}
			V_{\mu\nu} = \frac{1}{2}\nabla_{\mu}\phi\nabla_{\nu}\phi - \Lambda(\phi)g_{\mu\nu}.
			\label{Eq:Gao_Vdd}
		\end{align}
	\end{subequations}
	This reduction highlights that the model of Gao et al. corresponds to a particular choice of coupling within our framework, thereby illustrating that the non-Lagrangian character of such models is naturally accommodated within our broader constitutive framework.

	\paragraph{Our coupling choices.} To explore this framework further we must now make a particular choice for the tensor $V_{\mu\nu}$. To this end we define a potential function $V(\phi,\mathcal{X})$, which we often refer to as the ``coupling potential''. Given such a function, we write
	\begin{align}
		V_{\mu\nu} = V(\phi,\mathcal{X})g_{\mu\nu}.
	\end{align}
	This choice was motivated by, and is consistent with, both a minimally- and $k$-essence-coupled scalar field (provided the equation of motion is chosen appropriately), where the potential part of the energy-momentum tensor naturally takes a pure-pressure form proportional to the metric.

	\subsection{The $(3+1)$-decomposition and evolution equations}
	\label{Sec:The_(3+1)-decomposition_and_evolution_equations}
	We have now introduced our basic method for coupling matter fields to gravity and have made a specific choice of the potential part of the energy momentum tensor. However, in order to fully close the system given by \Eqref{Eq:Div_Tdd} it remains for us to pick some piece of the ``kinetic part'' $\T_{\mu\nu}$. There is of course no clear physically preferable way to do this. To discuss this further we make use of the $(3+1)$-decomposition. In \Sectionref{Sec:3+1_decomp_adm} we recall some basic facts about the $(3+1)$-decomposition and derive the ADM equations. The $(3+1)$-decomposition of the conservation equations \Eqref{Eq:Div_Tdd} and the equation of motion \Eqref{Eq:EoM} is then described in \Sectionref{Sec:3+1_decomp_div}. This allows us to write \Eqsref{Eq:EFEs}--\eqref{Eq:EoM} as an initial value problem and lets us identify which components of $\T_{\mu\nu}$ should be specified. We emphasize that the choices made here are not unique and should be understood as a first suggestion only. We finish this subsection in \Sectionref{Sec:Well-posedness_and_restrictions_on_the_free_data} where we discuss well-posedness.   
	
	\subsubsection{The $(3+1)$-decomposition of the Einstein equations}
	\label{Sec:3+1_decomp_adm}
	Given a spacetime $(M,g_{\alpha\beta})$, which is a solution of the Einstein equations \Eqsref{Eq:EFEs}, we suppose that there exists a smooth function $t:M\rightarrow \mathbb{R}$ whose collection of level sets $\Sigma_t$ forms a foliation $\Sigma$ of $M$. This foliation yields a decomposition of $(M,g_{\alpha\beta})$ in the standard way. The unit co-normal of any $3$-surface $\Sigma_t\in\Sigma$ is 
	\begin{subequations}
		\begin{align}
			n_{\mu}=\alpha \nabla_{\mu}t,
			\label{eq:def_Na}
		\end{align}
		where $\alpha>0$ is the \emph{lapse}. The induced first and second fundamental forms are therefore,
		\begin{align}
			\gamma_{\mu\nu}=g_{\mu\nu}+n_{\mu}n_{\nu},
			\label{Eq:Metriic_Decomp}
		\end{align}
		and
		\begin{align}
			K_{\mu\nu}=-\frac{1}{2}\mathcal{L}_{n}\gamma_{\mu\nu},
		\end{align}
		respectively. Here, $\mathcal{L}_n$ represent the Lie derivative in the direction of $n^\mu$. The covariant derivative associated with $\gamma_{\mu\nu}$ is $D_{\alpha}$. The tensor field
		\begin{align}
			{\gamma^\mu}_{\nu}={\delta^\mu}_{\nu}+n^{\mu}n_{\nu},
		\end{align}
	\end{subequations}
	is the map that projects any tensor defined at any point in $M$ orthogonally to a tensor that is tangent to some $\Sigma_t$. If each index of a tensor field defined on $M$ contracts to zero with $n_{\mu}$ or $n^{\mu}$, then we call that field \textit{spatial}. Given an arbitrary tensor field on $M$ we can create a spatial tensor field on $\Sigma_t$ by contracting each index with ${\gamma^\alpha}_\beta$. In fact, any tensor can be uniquely decomposed into its intrinsic and its orthogonal parts, e.g.
	\begin{subequations}
		\begin{align}
			T_{\mu\nu}=\rho n_{\mu}n_{\nu}+n_{\mu}j_{\nu}+n_{\nu}j_{\mu}+S_{\mu\nu},
		\end{align}   
		where
		\begin{align}
			\label{eq:Tdec}
			\rho=n^{\nu}n^{\mu}T_{\mu\nu},
			\quad
			j_{\nu} = -{\gamma^\sigma}_{\nu}n^\mu T_{\mu\sigma},
			\quad
			S_{\mu\nu} = {\gamma^\sigma}_{\mu}{\gamma^\eta}_{\nu}T_{\sigma\eta}.
		\end{align}
		The equations in this subsection focus on the geometric part of the evolution only. From the perspective of the geometry, the fields $(\rho,j_\mu,S_{\mu\nu})$ appear as source terms only and as such that can be discussed separately. We shall discuss these fields in more detail in the next subsection.
	\end{subequations}
	
	Now pick an arbitrary vector field  $t^\mu$ such that 
	\begin{subequations}
		\begin{align}
			t^{\mu}D_{\mu}t = 1.
		\end{align}
		According to \Eqref{eq:def_Na} there must exist a unique spatial vector field $\beta^{\mu}$, called the \textit{shift}, such that
		\begin{align}
			t^\mu = \alpha n^{\mu} + \beta^{\mu}.
			\label{Eq:n_t}
		\end{align}	
	\end{subequations}
	Collectively, the lapse $\alpha$ and shift $\beta_\mu$ are commonly refereed to as the \emph{gauge}. 
	
	Given all this, we find that $K_{\mu\nu}$ and $\gamma_{\mu\nu}$ are solutions of the evolution equations
	\begin{subequations}
		\begin{align}
			\mathcal{L}_{n}{K}_{\mu\nu}=& -\alpha^{-1}D_{\mu}D_{\nu}\alpha + \left( {R}_{\mu\nu} + K{K}_{\mu\nu} - 2{K^\sigma}_{\mu}K_{\sigma\nu} \right) + \frac{1}{2}\left( (S-\rho){\gamma}_{\mu\nu} - 2 {S}_{\mu\nu} \right),
			\label{Eq:Evol_K}
			\\
			\mathcal{L}_{n}\gamma_{\mu\nu}=&-2{K}_{\mu\nu},
			\label{Eq:Evol_y}
		\end{align}
		where ${R^{\mu}}_\nu$ is the Ricci tensor associated with $\gamma_{\mu\nu}$, and where $K=\gamma^{\mu\nu}K_{\mu\nu}$ is the \emph{mean curvature}. Note that \Eqref{Eq:Evol_K} follows from the fully spatial projection of \Eqref{Eq:EFEs}, while \Eqref{Eq:Evol_y} is nothing more than the definition of the extrinsic curvature. Observe that \Eqsref{Eq:Evol_K}--\eqref{Eq:Evol_y} are the standard \emph{ADM evolution equations}. In addition we also get the \emph{ADM constraint equations} as
		\begin{align}
			2\rho - ( R + K^{2} - K_{\mu\nu}K^{\mu\nu} ) =&0,
			\label{Eq:Evol_Theta}
			\\
			j_{\mu}+ D^{\nu}K_{\mu\nu} - D_{\mu}K=&0,
			\label{Eq:Evol_X}
		\end{align}
	\end{subequations}
	where $R$ is the Ricci scalar. Note that \Eqsref{Eq:Evol_Theta} and \eqref{Eq:Evol_X} follow from the fully normal and mixed projections of \Eqref{Eq:EFEs}, respectively. \Eqsref{Eq:Evol_Theta} and \eqref{Eq:Evol_X} are often referred to as the Hamiltonian and momentum constraints, respectively.
	
	Observe that all quantities here  are smooth \emph{spatial} tensor fields. It is clear that while this means that all contractions with $n^\sigma$ or $n_\sigma$ vanish, contractions with $t^\sigma$ do not, e.g., $j_t:=j_\sigma t^\sigma=j^\sigma t_\sigma$ as a consequence of \Eqref{Eq:n_t}. However such ``components'' $j_t$ do clearly not constitute a further degree of freedom of the field $j_\sigma$ since $j_t=j_\sigma \beta^\sigma$ is fully determined by its ``spatial components''. Consistent with this, one easily checks that the equation for $j_t$ --obtained by contracting \Eqref{Eq:Evol_X} with $t^\sigma$-- fully decouples from the remaining equations. We remark that instead of thinking of each field in the Einstein equations (\Eqsref{Eq:Evol_K}--\eqref{Eq:Evol_X}) above as a spatial field on $M$, we could equivalently think of it as a $1$-parameter family of fields on $\Sigma_t$ defined by the pull-back along the $t$-dependent map $\Lambda_t: \Sigma_t\rightarrow M$, $p\mapsto (t,p)$ to $\Sigma_t$. In the following we shall use abstract indices  $a,b,\ldots$ for such $t$-dependent tensor fields  on $\Sigma_t$. Indeed, all indices $\mu,\nu,\ldots$ in the Einstein equations above could be replaced by $a,b,\ldots$, and, at the same time, each Lie-derivative along $t^\sigma$ by the derivative with respect to parameter $t$. All this is well-known for $(3+1)$-decompositions of spacetimes and is therefore not discussed any further here.

	We end with a brief discussion of well-posedness. Whether or not \Eqsref{Eq:Evol_K}--\eqref{Eq:Evol_X} lead to a well-posed initial value problem depends on the particular choice of the \emph{gauge} $(\alpha,\beta_\mu)$. It is well known, for example, that the Constant Mean Curvature (CMC) gauge with zero shift ensures that \Eqsref{Eq:Evol_K}--\eqref{Eq:Evol_X} has a well-posed initial value problem \cite{RodnianskiSpeck:Linear,RodnianskiSpeck:NonLinear}. At this stage, we do not make any particular choice of gauge $(\alpha,\beta_\mu)$ as we wish to leave our discussion as general as possible. Thus, for the remainder of this section, we shall assume that the gauge has been chosen in such a way as to ensure that \Eqsref{Eq:Evol_K}--\eqref{Eq:Evol_X} have a well-posed Cauchy problem. Given this, it only remains to show that the matter sector --described by the fields $(\phi,\rho,j_\mu,S_{\mu\nu})$-- also have a well-posed initial value problem. 
	
	\subsubsection{The $(3+1)$-decomposition of the conservation equations}
	\label{Sec:3+1_decomp_div}
	We now use the structure provided by the $(3+1)$-decomposition, described above, to discuss the kinetic part of the energy momentum tensor $\T_{\mu\nu}$. To this end we write  
	\begin{subequations}
		\begin{align}
			\T_{\mu\nu}=\kappa n_{\mu}n_{\nu} + j_{\mu}n_{\nu} + n_{\mu}j_{\nu} + q_{\mu\nu},
		\end{align}
		where 
		\begin{align}
			\kappa = n^{\sigma}n^{\eta}\T_{\sigma\eta},
			\quad
			j_{\mu} = -{\gamma^\sigma}_{\mu}n^\eta\T_{\sigma\eta},
			\quad
			q_{\mu\nu} = {\gamma^\sigma}_{\mu}{\gamma^\eta}_{\nu}\T_{\sigma\eta}.
			\label{eq:T_dec}
		\end{align}
		Observe carefully that the mixed projections of $T_{\mu\nu}$ and $\T_{\mu\nu}$ are the same since ${\gamma^\mu}_{\sigma}n^\nu g_{\mu\nu}=0$. Comparing now \Eqref{eq:T_dec} to \Eqref{eq:Tdec} we find 
		\begin{align}
			\rho = \kappa + V(\phi,\mathcal{X}),
			\quad
			S_{\mu\nu} = q_{\mu\nu} - V(\phi,\mathcal{X})\gamma_{\mu\nu}. 
		\end{align}	
		For later convenience, we write $q_{\mu\nu}$ in two parts:  
		\begin{align}
			q_{\mu\nu}=Q_{\mu\nu}+\frac{1}{3}q\gamma_{\mu\nu},
			\quad
			Q = \gamma^{\mu\nu}Q_{\mu\nu},
		\end{align}
	\end{subequations}
	where $Q_{\mu\nu}$ is an arbitrary symmetric tensor (i.e., $Q_{\mu\nu}=Q_{\nu\mu}$), and $q$ is a freely specifiable function. Given all of this, we find that the \emph{conservation equations} \Eqref{Eq:Div_Tdd} can be written as 
	\begin{subequations}
		\begin{align}
			\mathcal{L}_{n}\kappa + D^{\mu}j_{\mu}   =\, & \frac{3\kappa + q}{3} K - 2\alpha^{-1} j^\mu D_{\mu}\alpha +  Q_{\mu\nu}K^{\mu\nu}  - \mathcal{L}_{n}V(\phi,\mathcal{X}),
			\label{Eq:Evol_kappa_qdd}
			\\
			\mathcal{L}_{n}j_{\sigma} + \frac{1}{3}D_{\sigma}{q} =\, &  D_{\sigma}V(\phi,\mathcal{X}) - \frac{3\kappa + q}{3\alpha}D_{\sigma}\alpha + K j_{\sigma} - D^{\mu}Q_{\mu\sigma} - \alpha^{-1}Q_{\mu\sigma}D^{\mu}\alpha.
			\label{Eq:Evol_pd_qdd}
		\end{align}
	\end{subequations}
	This system ``naturally'' forms a set of evolution equations for the variables $\kappa$ and $j_\sigma$. However, in order to solve \Eqsref{Eq:Evol_kappa_qdd}--\eqref{Eq:Evol_pd_qdd} one must first make specific choices of the function $q$ and the tensor $Q_{\mu\nu}$.  
	
	Turning our attention to the equation of motion for $\phi$, \Eqref{Eq:EoM}, we write
	\begin{subequations}
		\begin{align}
			\nabla_{\alpha}\phi = -\nu n_{\alpha} + w_\alpha,
			\quad
			\tilde{B}^\alpha = -r n^\alpha + \hat{r}^{\alpha},
		\end{align}
		where 
		\begin{align}
			\nu = n^{\sigma}\nabla_{\sigma}\phi,
			\quad
			w_\alpha = {\gamma^{\sigma}}_{\alpha}\nabla_{\sigma}\phi,
			\quad
			{r}=n_{\sigma}\tilde{B}^\sigma,
			\quad
			\hat{r}^{\alpha} = {\gamma^\alpha}_\sigma \tilde{B}^\sigma.
		\end{align}
	\end{subequations}
	In terms of the variables $(\nu,w_\alpha)$ we find that $\nabla_{\alpha}\nabla_\beta\phi$ can be written as 
	\begin{align}
		\begin{split}
			\nabla_{\alpha}\nabla_\beta\phi = \left( \mathcal{L}_n \nu - \omega_{\nu}a^\nu \right)n_\alpha n_\beta &- \left( \mathcal{L}_{n}w_\beta + w^\sigma K_{\sigma\beta} - \nu a_\sigma \right)n_{\alpha} 
			\\
			&- \left( D_{\alpha}\nu + w^{\sigma}K_{\sigma\alpha} \right)n_{\beta}+\nu K_{\alpha\beta} + D_{\mu}w_{\nu},
		\end{split}
		\label{Eq:Hessian_decomp}
	\end{align}
	Note here that smoothness of $\phi$ implies $\nabla_{\alpha}\nabla_{\beta}\phi = \nabla_{\beta}\nabla_{\alpha}\phi$. This requirement leads to an evolution equation for $w_\sigma$, which can be interpreted as an integrability condition. This equation is given below, in \Eqref{Eq:IntegrabilityCondition}. Turning our attention to the matrix $B^{\alpha\beta}$ we write
	\begin{subequations}
		\begin{align}
			B^{\alpha\beta} = \mathring{b}\left( n^{\alpha}n^{\beta} - 2\hat{b}^{\left( \alpha\right.}n^{\left. \beta\right)} + b^{\alpha\beta}\right),
			\label{Eq:BUU_decomp}
		\end{align} 
		with 
		\begin{align}
			\mathring{b} = n_{\alpha}n_{\beta} B^{\alpha\beta},
			\quad
			\hat{b}^{\alpha} = -(\mathring{b})^{-1}{\gamma^\alpha}_\sigma n_\beta B^{\sigma\beta},
			\quad
			b^{\alpha\beta} = (\mathring{b})^{-1}{\gamma^\alpha}_\sigma{\gamma^\beta}_\iota B^{\sigma\iota}.
		\end{align}
	\end{subequations}
	Contracting the decomposed Lorentz metric \Eqref{Eq:BUU_decomp} with \Eqref{Eq:Hessian_decomp} (and the vector $\tilde{B}^\mu$ with $\nabla_\mu\phi$) allows us to derive an evolution equation for $\nu$ in terms of $(3+1)$ variables. The resulting first order system is 
	\begin{subequations}
		\label{Eq:EOM_Gen_3+1}
		\begin{align}
			&\mathcal{L}_{n}\phi = \nu,
			\\ 
			&\mathcal{L}_{n}w_{c} = D_{c}\nu + \nu a_{c}, 
			\label{Eq:IntegrabilityCondition}
			\\
			\begin{split}
				&\mathcal{L}_{n}\nu={\mathring{b}}^{-1}\left( r\nu -\hat{r}^\alpha w_\alpha + f(\phi,\mathcal{X}) \right) + w^\sigma a_\sigma -2 \hat{b}^\alpha D_{\alpha}\nu  - b^{\alpha\beta}D_{\alpha}w_{\beta}
				\\
				&\quad\quad\quad- (2\hat{b}^{\alpha}w^{\beta} + \nu b^{\alpha\beta})K_{\alpha\beta}
			\end{split}
		\end{align}
	\end{subequations}
	We note here that, since $B^{\alpha\beta}$ is a symmetric Lorentz type metric, \Eqref{Eq:EOM_Gen_3+1} is well-posed.
	
	\subsubsection{Well-posedness and restrictions on the free data}
	\label{Sec:Well-posedness_and_restrictions_on_the_free_data}
	As stated above, in the present work, we assume	that $\phi$ is determined by \Eqref{Eq:EoM}. However, this equation is phenomenological, and it is certainly possible for $\phi$ to satisfy a much more general evolution equation. The precise choice of the equation of motion can, in principle, affect the well-posedness properties of \Eqsref{Eq:Evol_kappa_qdd}--\eqref{Eq:Evol_pd_qdd}. The purpose of this subsection	is to discuss exactly this issue.
	
	For our discussion here the fields $B^{\alpha\beta}$ and $\tilde{B}^\beta$ may depend on $\kappa$ or $j_\sigma$ but \emph{not} their derivatives\footnote{We emphasise that more general assumptions are of course possible. The choices made here ensure that the principal part of \Eqref{Eq:EoM} decouples from the the principal part of \Eqsref{Eq:Evol_kappa_qdd}--\eqref{Eq:Evol_pd_qdd}. If $B^{\alpha\beta}$ and $\tilde{B}^\beta$ are allowed to depend on derivatives of $\kappa$ or $j_\sigma$ then the principal parts do not decouple and care should be taken to ensure that the resulting system is well-posed.}. This means that the principal part of \Eqref{Eq:EoM} effectively decouples from the principal part of \Eqsref{Eq:Evol_kappa_qdd}--\eqref{Eq:Evol_pd_qdd}, and hence from the
	perspective of \Eqsref{Eq:Evol_kappa_qdd}--\eqref{Eq:Evol_pd_qdd}, the scalar field
	$\phi$ may be regarded as a \emph{given function}.
	
	Turning now to the evolution equations for $\kappa$ and $j_\sigma$, we suppose that $Q_{\mu\nu}$ \emph{does not} depend on $\kappa$ or $j_\mu$. It may however, possess a dependence on $\phi,\nu$ or $w_\sigma$ (but not their derivatives). In addition the function $q$ is assumed to depend on $\kappa$ \emph{only} (i.e., $q=q(\kappa)$). Given these assumptions we find that \Eqref{Eq:Evol_pd_qdd} can now be written as 
	\begin{align}
		\begin{split}
			\mathcal{L}_{n}j_{\sigma} + \frac{1}{3}q^\prime(\kappa)D_{\sigma}\kappa =\, &  D_{\sigma}V(\phi,\mathcal{X}) - \frac{3\kappa + q}{3\alpha}D_{\sigma}\alpha + K j_{\sigma} - D^{\mu}Q_{\mu\sigma} - \alpha^{-1}Q_{\mu\sigma}D^{\mu}\alpha.
		\end{split}
		\label{Eq:Evol_pd_qdd_2}
	\end{align}
	with ${q}^\prime=d{q}/d\kappa$. Note that the principal part of \Eqref{Eq:Evol_kappa_qdd} is unchanged. It turns out that \Eqsref{Eq:Evol_kappa_qdd} and \eqref{Eq:Evol_pd_qdd_2} form a symmetric hyperbolic system with symmetrizer 
	\begin{subequations}
		\begin{align}
			\frac{1}{3}\left(
			\begin{array}{cc}
				3 & 0 \\
				0 & {q}^{\prime}(\kappa){\gamma^{\sigma}}_{\nu}
			\end{array}
			\right),
		\end{align}
		provided
		\begin{align}
			{q}^{\prime}(\kappa)\ge0,
			\label{Eq:Hyperbolicity}
		\end{align} 
	\end{subequations}
	where $\gamma^{\sigma\mu}$ is the spatial inverse of $\gamma_{\sigma\mu}$ so that $\gamma^{\mu\sigma}\gamma_{\sigma\nu}=\text{diag}(0,1,1,1)$. Physically, this condition ensures that the effective speed of sound of the scalar medium $c_s^2=q^\prime(\kappa)/3$ remains real, thereby preventing the catastrophic growth of high-frequency instabilities. We refer to \Eqref{Eq:Hyperbolicity} as the \emph{hyperbolicity condition}. It follows that for arbitrary free data, for which the hyperbolicity condition \Eqref{Eq:Hyperbolicity} holds, \Eqsref{Eq:Evol_kappa_qdd} and \eqref{Eq:Evol_pd_qdd_2} is a quasilinear strongly-hyperbolic system and the \emph{Cauchy problem} in both the increasing \emph{and} decreasing $t$-directions is well-posed (at least locally). Note that, in the special case $q^\prime(\kappa)=0$, \Eqsref{Eq:Evol_kappa_qdd} and \eqref{Eq:Evol_pd_qdd_2} decouple in leading order. If we also have that $D_\mu\alpha=0$ (as in the case for geodesic slicing) then the equations completely decouple and can therefore be solved consecutively (i.e., one would solve \Eqref{Eq:Evol_pd_qdd_2} and then \Eqref{Eq:Evol_kappa_qdd}). However, if $D_{\mu}\alpha\neq0$ then the system is only weakly hyperbolic.

	Given all of this, \Eqsref{Eq:Evol_kappa_qdd} and \eqref{Eq:Evol_pd_qdd_2} now suggest the following groupings: 
	\begin{description}		
		\item[Free Data:] The symmetric tensor ${Q}_{ab}$ and the function $V(\phi)$ are freely specifiable functions of $\phi$ (and its derivatives) everywhere on $M$ provided that $V(\phi)$ depends only on $\phi$ and $Q_{ab}$ does not depend on the unknowns. Similarly the field ${q}$ is a freely specifiable function of $\kappa$ everywhere on $M$.
		\item[Equation of state:] The field $q(\kappa)$, which is a freely specifiable everywhere on $M$ subject to the condition $q^\prime(\kappa)\ge 0$, plays the role of an \emph{equation of state}, similar to what is seen in constitutive theories such as hydrodynamics.		
		\item[Unknowns:] The scalar and vector fields $\kappa$ and $j_{\sigma}$ are the unknowns. Given appropriate initial	data, the task is to determine these as solutions of \Eqsref{Eq:Evol_kappa_qdd}--\eqref{Eq:Evol_pd_qdd}. Note here that, given appropriate free fields, all coefficients in these equations are determined as functions of the scalar field and its derivatives everywhere on $M$.
	\end{description}
	
	Note that one could allow $q$ to also depend on $\phi$ (or its derivatives). The resulting calculation is nearly identical to the one above and the corresponding hyperbolicity condition is $\partial q/\partial\kappa>0$. This is of course more general, however, the assumption $q=q(\kappa)$ is sufficient for our purposes here. Moreover, if $q$ is allowed to depend on $j_\sigma$ the resulting system could still be hyperbolic. However, such a dependence would alter the principal part of the equations and hence care should be taken to ensure that the system remains hyperbolic.

	\subsection{Reductions}
	\label{Sec:Reductions}
	We now discuss some particular choices for the free fields. The purpose of these examples is to demonstrate that, under appropriate restrictions, a number of standard scalar field models are contained within our formalism. In particular, we show that both the minimally coupled scalar field and the k-essence scalar field --discussed in \Sectionref{Sec:Comparison_to_a_minimally_coupled_scalar_field} and \Sectionref{Sec:Reduction_to_k-essence_scalar_fields}, respectively-- arise as special cases. The examples presented here will form the basis for our subsequent discussion of the role played by the free fields $(q(\kappa),Q_{ab})$ in \Sectionref{Sec:Role_of_the_free_data}.

	\subsubsection{Reduction to a minimally coupled scalar field}
	\label{Sec:Comparison_to_a_minimally_coupled_scalar_field}
	
	We now show that the framework introduced above reduces to the standard case of a minimally coupled scalar field. In particular, we demonstrate that for a suitable choice of free data, the evolution system \Eqsref{Eq:Evol_kappa_qdd}--\eqref{Eq:Evol_pd_qdd_2} admits solutions corresponding to the usual scalar-field energy-momentum tensor (see 	\Eqref{Eq:MinimallyCoupled}). To this end, we pick the fields $B^{\mu\nu},\tilde{B}^\mu$, and $f(\phi,\mathcal{X})$ as in \Eqref{Eq:MinimalFreeData}. For these choices \Eqsref{Eq:EOM_Gen_3+1} reduce to
	\begin{subequations}
		\begin{align}
			\mathcal{L}_{n}\nu = D_{a}w^a + w^c \alpha^{-1}D_c \alpha + \nu K - V^\prime(\phi), 
			\quad
			\mathcal{L}_{n}w_{c} = D_{c}\nu + \nu a_{c}, 
			\quad
			\mathcal{L}_{n}\phi = \nu.
			\label{Eq:BoxDecomp}
		\end{align}
		The equations of motion given in \Eqref{Eq:BoxDecomp} are not yet sufficient to close the system \Eqsref{Eq:Evol_kappa_qdd} and \eqref{Eq:Evol_pd_qdd_2}. For this we must first make specific choices of the free data. In this subsection we set
		\begin{align}
			{q}(\kappa) = 3\kappa,
			\quad
			Q_{ab} = w_{a}w_{b} -  w_{c} w^{c} \gamma_{ab} .
			\label{Eq:FreeData_Q_q}
		\end{align}
	\end{subequations}
	Notice that ${q}^\prime(\kappa)=3>0$, and hence the hyperbolicity condition
	\Eqref{Eq:Hyperbolicity} is satisfied. Moreover, given these choices, and assuming
	that $(\phi,\nu,w_a)$ satisfies \Eqref{Eq:BoxDecomp}, we find that
	\begin{align}
		\kappa = \frac{1}{2}\left( w_{a} w^{a} + \nu^2 \right),
		\qquad
		j_a = -\nu w_a,
		\label{Eq:ParticualSol_minimal}
	\end{align}
	is a \emph{particular} solution of \Eqsref{Eq:Evol_kappa_qdd}--\eqref{Eq:Evol_pd_qdd_2}. It follows that, for this choice of free data, the variables $(\kappa,j_a)$ reproduce the standard $(3+1)$-decomposition of the energy-momentum tensor of a minimally coupled scalar field.
	
	\subsubsection{Reduction to $k$-essence scalar fields}
	\label{Sec:Reduction_to_k-essence_scalar_fields}
	
	For our second example, we show that the framework introduced above contains $k$-essence type scalar field models. In particular, we demonstrate that for a suitable choice of free data, the evolution system \Eqsref{Eq:Evol_kappa_qdd}--\eqref{Eq:Evol_pd_qdd_2} admits solutions corresponding to the usual $k$-essence scalar-field energy-momentum tensor. For this, we suppose that the fields $B^{\mu\nu},\tilde{B}^\mu$, and $f(\phi,\mathcal{X})$ are chosen as
	\begin{subequations}
		\begin{align}
			B^{\mu\nu} = -\frac{\partial V}{\partial\mathcal{X}} g^{\mu\nu},
			\quad
			\tilde{B}^\mu = -\nabla^{\mu}\left( \frac{\partial V}{\partial\mathcal{X}} \right),
			\quad
			f(\phi,\mathcal{X})=\frac{\partial V}{\partial\phi}.	
		\end{align}
		Note that, for the sake of readability, we have suppressed the arguments of $V=V(\phi,\mathcal{X})$. In this case, the equation of motion \Eqref{Eq:EoM} can be written as 
		\begin{align}
			\nabla_{\mu}\left( \frac{\partial V}{\partial\mathcal{X}}\nabla^\mu\phi \right) = \frac{\partial V}{\partial\phi}
			\label{Eq:Box_q}
		\end{align}
		As before, this equation alone is not sufficient to close \Eqsref{Eq:Evol_kappa_qdd} and \eqref{Eq:Evol_pd_qdd_2}. For the free data, we pick 
		\begin{align}
			q(\kappa) = 3\kappa, 
			\quad
			Q_{ab} = \left( \nu^2 \gamma_{ab} - w_{a}w_{b}  \right)\frac{\partial V}{\partial\mathcal{X}}
		\end{align}
	\end{subequations}
	Notice that ${q}^\prime(\kappa)=3>0$, and hence the hyperbolicity condition
	\Eqref{Eq:Hyperbolicity} is satisfied. Moreover, given these choices, and assuming
	that $(\phi,\nu,w_a)$ satisfies \Eqref{Eq:Box_q}, we find that
	\begin{align}
		\kappa = -\nu^2 \frac{\partial V}{\partial\mathcal{X}},
		\qquad
		j_a = \nu \frac{\partial V}{\partial\mathcal{X}}w_{a},
		\label{Eq:ParticualSol_k}
	\end{align}
	is a \emph{particular} solution of \Eqsref{Eq:Evol_kappa_qdd} and \eqref{Eq:Evol_pd_qdd_2}. It follows that, for this choice of free data, the variables $(\kappa,j_a)$ reproduce the standard $(3+1)$-decomposition of the energy-momentum tensor of a $k$-essence scalar field.
	
	\subsection{Role of the free data}
	\label{Sec:Role_of_the_free_data}
	We now discuss the fields $(q(\kappa),Q_{ab})$. For this, we begin by noting that in \Sectionref{Sec:Comparison_to_a_minimally_coupled_scalar_field} (and in \Sectionref{Sec:Reduction_to_k-essence_scalar_fields}) we assumed that $q(\kappa)=3\kappa$. This is not the only possible choice of $(q(\kappa),Q_{ab})$ that gives rise to the solution \Eqref{Eq:MinimallyCoupled} (or \Eqref{Eq:ParticualSol_k}). This is a consequence of the fact that $(\kappa,j_a)$ is determined \emph{before} the free data has been chosen. As a result, for \emph{any} choice of $q(\kappa)$ (which satisfies the hyperbolicity condition \Eqref{Eq:Hyperbolicity}), there exists some choice of $Q_{ab}$ that leads to the solution \Eqref{Eq:ParticualSol_minimal}. This observation is generic. To understand why, suppose that $(\mathring{\kappa},\mathring{j}_a)$ is a known solution of \Eqsref{Eq:Evol_kappa_qdd} and \eqref{Eq:Evol_pd_qdd_2} corresponding to the free fields $(q(\kappa),Q_{ab})$. Then, $(\mathring{\kappa},\mathring{j}_a)$ is \emph{also} a solution corresponding to the free fields $(\hat{q}(\kappa),\hat{Q}_{ab})$ where
	\begin{align}
		\hat{Q}_{ab} = Q_{ab} + \frac{1}{3} (q(\mathring{\kappa})-\hat{q}(\mathring{\kappa}))\gamma_{ab}.
		\label{Eq:ConstituativeFreedom}
	\end{align}
	In other words, different choices of the function $q(\kappa)$ can be compensated for by altering the choice of $Q_{ab}$ \emph{without} changing the background solution $(\mathring\kappa,\mathring{j}_a)$. This freedom should be understood as part of the definition of the matter model rather than as an ambiguity in the solution itself. The quantities $q(\kappa)$ and $Q_{ab}$ do not affect whether a given configuration satisfies the field equations, but they do determine how the matter variables $(\kappa,j_a)$ respond dynamically. In this sense, the free data play a role analogous to constitutive relations in continuum matter models, such as the Navier-Stokes equations; there, the conservation of momentum is a universal law, but the specific behaviour of the fluid --whether it is water, honey, or a non-Newtonian substance-- is determined by the constitutive choice of the viscosity tensor. In our case, $q(\kappa)$ and $Q_{ab}$ characterize the ``constitutive response'' of the scalar medium. One immediate consequence of this observation is that properties such as linear stability are not intrinsic properties of the background configuration alone. Although the background solution is unchanged under variations of the free data, the corresponding linearized evolution equations for $(\kappa,j_\sigma)$ depend explicitly on the choice of $(q(\kappa),Q_{ab})$. As a result, the characteristic structure, and hence the stability properties, may differ for different admissible choices of free data, even when evaluated about the same background solution. In standard variational models, these degrees of freedom are ``frozen'' by the choice of a Lagrangian. In the case of a minimally coupled scalar field (or $k$-essence scalar field), the constitutive relations are fixed by the requirement that $\kappa$ and $j_a$ be algebraically determined in terms of $(\phi,\nu,w_a)$. This removes the freedom described above and uniquely ties the dynamical response of the matter sector to the scalar-field evolution. From the present perspective, these models therefore correspond to highly constrained special cases within a broader class of admissible matter models, rather than to a generic outcome of the conservation equations alone. 
	
	\section{Bianchi I solutions for ``near-minimal'' scalar fields}
	\label{Sec:Bianchi_I}
	We have seen now that variational models (such as the minimally coupled scalar field) are contained within our formalism subject to some algebraic constraint. The purpose of this section here is to explore this further. To that end we suppose that free data and the equation of motion are somehow  ``close'' to a minimally coupled scalar field (the details of our choices are discussed in \Sectionref{Sec:Our_free_data_choices}). In addition, we restrict our attention to the class of spatially homogeneous Bianchi I solutions. This class of cosmologies is amongst the simplest possible type of spacetimes and is therefore an appropriate setting for our initial explorations here. A derivation of the Bianchi I equations is presented in \Sectionref{Sec:The_Bianchi_1_equations}. Solutions of the Bianchi I system are then discussed in \Sectionref{Sec:Bianchi_1_solutions}. In \Sectionref{Sec:Kasner_Solutions} we first suppose we have a minimally coupled scalar field with zero potential and derive the \emph{Kasner solutions}. These are amongst the simplest known anisotropic cosmologies and play a key role in mathematical cosmology largely due to their role in the BKL conjecture and stable Big Bang formation. In this setting we are able to explicitly explore how the algebraic constraints \Eqref{Eq:ParticualSol_minimal} appear in the exact solutions. 
	
	\subsection{Our free data choices}
	\label{Sec:Our_free_data_choices}
	As stated above the goal of this section is to investigate scalar field solutions which are constructed via our formalism. For this we must now make specific choices for the various free fields. We do this in two steps. First, we make a specific choice of topology underlying the foliation $\Sigma=\{\Sigma_t\}$ and introduce a coordinate system that is adapted to this foliation. Our choices here are discussed in \Sectionref{Sec:Topological_considerations}. Second, we make specific choices of the free fields $Q_{ab}$ and $q(\kappa)$ as well as on the fields $B^{\alpha\beta},\tilde{B}^\alpha,f(\phi)$, and $V(\phi)$. These choices are discussed in \Sectionref{Sec:Constitutive_freedoms}.
	
	\subsubsection{Topological considerations and gauge choices} 
	\label{Sec:Topological_considerations}
	We begin by specifying the underlying topology of the space-times $(M,g_{\mu\nu})$ studied in this work. To that end, suppose now we have chosen a smooth time function $t:M\rightarrow\mathbb{R}$, with the properties discussed above, giving rise to a foliation $\Sigma=\{\Sigma_t\}$ whose level sets are diffeomorphic to some compact $3$-surface $\tilde{\Sigma}$. i.e.,
	\begin{subequations}
		\begin{align}
			M = I\times \tilde{\Sigma},
			\label{Eq:M_topology}
		\end{align}
		for some time-interval $I$. We write the points in the foliation $\Sigma$ as $(t,p)$ with $t\in I$ and $p\in\tilde{\Sigma}$. Observe carefully that we often use the same symbol $t$ for the real parameter $t\in I$ as well as for the \emph{function} $t(p)$ that defines the $(3+1)$-decomposition. In some sense the particular choice of topology for $\tilde{\Sigma}$ is irrelevant as our focus in the section is on Bianchi I solutions only. In this setting, any of the unknown fields defined on $M$ can be equivalently thought of as existing on the $1$-dimensional interval $I\subset M$ with coordinate $t$. In this picture, given any point $p\in\tilde\Sigma$, one thinks of the fields defined on $I$ as the pull-back to $I$ along the ($p$-dependent) map $\Phi_p : I \rightarrow M, t\mapsto (t, p)$. $I$ therefore acts as the ``effective manifold'' of our analytical solutions. Here, we shall assume that there is some $T>0$ so that 
		\begin{align}
			I = \left(0,T\right].
		\end{align} 
	\end{subequations}
	In regards to the lapse and shift (which encode the coordinate freedoms between consecutive hypersurfaces \cite{Alcubierre:Book}) we adopt the CMC gauge with zero shift. In this setting we pick shift $\beta^i$ and mean curvature $K$ as
	\begin{subequations}
		\begin{align}
			K = -1/t,
			\quad
			\beta^{i} = 0,
			\label{Eq:shift}
		\end{align}
		while the lapse is determined as a solution of the PDE 
		\begin{align}
			\gamma^{ab}D_{a}D_{b}\alpha =\alpha\left( R + \frac{1}{t^2} + \frac{1}{2}(Q+q-3\kappa)-3V(\phi) \right)-\frac{1}{t^2}.
			\label{Eq:Evol_alpha}
		\end{align}
	\end{subequations}
	
	\subsubsection{Constitutive freedoms and the equation of motion}
	\label{Sec:Constitutive_freedoms}
	The fields $Q_{ab}$ and $q$ are considered as freely specifiable in \Eqsref{Eq:Evol_kappa_qdd} and \eqref{Eq:Evol_pd_qdd_2}. There is of course no clear physically or geometrically preferable way to choose these freedoms. Given that our interest in the section is on models that are somehow close to a minimally coupled scalar field we choose our free data as in \Eqref{Eq:FreeData_Q_q}. Moreover, in regards to $B^{\alpha\beta}$ and $\tilde{B}^\mu$ we pick these fields as in \Eqref{Eq:MinimalFreeData}. In this setting, the equation of motion is 
	\begin{align}
		g^{\mu\nu}\nabla_{\mu}\nabla_{\nu}\phi = f(\phi).
	\end{align}
	These choices therefore allow us to study a system that is similar to the standard minimally coupled scalar field, specifically utilizing the constitutive freedom identified in \Eqref{Eq:ConstituativeFreedom} to allow the forcing term $f(\phi)$ and the potential $V(\phi)$ to evolve independently.

	\subsection{The Bianchi I equations}
	\label{Sec:The_Bianchi_1_equations}
	We now write down the evolution equations (corresponding to the choices described above) in the special case of Bianchi I spacetimes. In this setting, all of the unknowns (i.e., $\phi,\nu,\kappa$, and the functional components of $\gamma_{ab}$ and $K_{ab}$) depend only on the time coordinate $t$. Given this assumption it is a consequence of the momentum constraint \Eqref{Eq:Evol_X} that $j_a=0$. This particular $j_a$ is also a solution of the conservation equation \Eqref{Eq:Evol_pd_qdd_2}. Moreover, in this setting, \Eqref{Eq:Evol_alpha} becomes an algebraic equation and can therefore be solved explicitly to get 
	\begin{align}
		\alpha = \frac{1}{1-t^2 V(\phi)}.
	\end{align}
	{Observe carefully} that we have used that, in this setting, $Q_{ab}=0\implies Q=0$, which is a direct consequence of the assumption that ${Q}_{ab}$ depends only on spatial derivatives of $\phi$ (see \Eqref{Eq:FreeData_Q_q}). Given this, it is now convenient to split $K_{ab}$ into its trace and trace-free parts. i.e, 
	\begin{subequations}
		\label{Eq:Bianchi1Reduction}
		\begin{align}
			{K^{a}}_{b}={\chi^{a}}_{b}-\frac{1}{3t} {\gamma^{a}}_{b},
			\quad
			\chi^a{}_{a}=0,
			\quad
			{K^{a}}_{a}=-{1}/{t}.
		\end{align}
		In regards to the metric, we suppose that there is a scalar function $\Omega:M\rightarrow \mathbb{R}$ and a tensor $\tilde{\gamma}_{ab}$ such that 
		\begin{align}
			\gamma_{ab}=\Omega^{2}\tilde{\gamma}_{ab},
			\quad
			\tilde{\gamma}_{ab} = \text{diag}( \gamma_{1},\gamma_{2},\gamma_{3} ).
			\label{Eq:SpecialMetric}
		\end{align}
	\end{subequations}
	Note that we have restricted our attention to diagonal metrics. It should be noted that this is not a significant restriction as, for spatially homogeneous solutions, one can always perform a (local) coordinate transformation so that the metric $\gamma_{ab}$ is diagonal. 
	
	Given all of this we find that the evolution equations for the first and second fundamental form can be written as   
	\begin{subequations}
		\label{Eqs:HomoEvolve}
		\begin{align}
			\partial_{t}\Omega = \frac{1}{3t}\alpha \Omega,
			\quad
			\partial_{t}\tilde{\gamma}_{ab}=-2\alpha\tilde{\gamma}_{ac}{{\chi^c}_{b}},
			\quad
			\partial_{t}{\chi^{a}}_{b}= -\frac{1}{t}\alpha {\chi^{a}}_{b}.
		\end{align}
		The evolution equations for $\tilde{\gamma}_{ab}$ and ${\chi^a}_b$ can be further reduced by setting 
		\begin{align}
			\gamma_{i} = \tilde{\gamma}_{i}\exp\left( -2{C_i}p(t) \right),
			\quad
			\chi^a{}_b = \chi(t)\text{diag}(C_1,C_2,C_3),
		\end{align}
		where, for each $i=1,2,3$, we have that $\tilde{\gamma}_{i},C_i\in\mathbb{R}$ are integration constants, and where the functions $p(t)$ and $\chi(t)$ are solutions of the evolution equations 
		\begin{align}
			\partial_{t}p(t) = \alpha \chi(t),
			\quad
			\partial_{t}\chi(t)=-\frac{1}{t}\alpha \chi(t).
			\label{Eq:p_and_u}
		\end{align}
		Since ${\chi^a}_b$ is trace-free, we must have 
		\begin{align}
			C_1 + C_2 + C_3 = 0.
			\label{Eq:C_i}
		\end{align}
		
		In addition to all this, we find that the evolution equations for $\nu$ and $\phi$ are  
		\begin{align}
			\partial_{t}\nu = -\frac{1}{t}\alpha \nu - \alpha f(\phi),
			\quad
			\partial_{t}\phi =  \alpha\nu.
			\label{Eq:nu_phi_eqs_Bianchi1}
		\end{align}
		The remaining evolution equation is 
		\begin{align}
			\partial_{t}\kappa = -\frac{2}{t}\alpha\kappa - \partial_{t}V(\phi),
		\end{align}
		and the Hamiltonian constraint is
		\begin{align}
			\begin{split}
				\kappa =& \frac{1}{3t^2} - \frac{1}{2}\underbrace{(C_1^2+C_2^2+C_3^2)}_{:=C^2}\chi^2 - V(\phi)= \frac{1}{3t^2} - \frac{1}{2}C^2\chi^2 - V(\phi),
			\end{split} 
			\label{Eq:Hamiltonian_Bianchi_1}
		\end{align}
		One can check via direct calculation that $\kappa$, as given by the Hamiltonian constraint, is always a solution of the evolution equation for $\kappa$ and hence $\kappa$ can be calculated \emph{after} $\chi$ and $\phi$ have been determined.
	\end{subequations}

	\subsection{Bianchi I solutions}
	\label{Sec:Bianchi_1_solutions}
	We now investigate solutions of the Bianchi I equations \Eqsref{Eqs:HomoEvolve}. In \Sectionref{Sec:Kasner_Solutions} we set $f(\phi)=V(\phi)=0$ and construct the corresponding solutions. In this setting our model coincides with the standard minimally coupled scalar field model and produces the well-known Kasner spacetimes if and only if the algebraic constraint 
	\begin{align}
		\kappa=\frac{1}{2}\nu^2,
		\label{Eq:AlgebraicConstraint}
	\end{align}
	is imposed. If this constraint is \emph{not} satisfied then the solutions represent a generalisation of the Kasner scalar field spacetimes. The case of $V(\phi),f(\phi)\neq 0$ is discussed in \Sectionref{Sec:NearMinimal_Kasner_Solutions}, wherein we establish conditions on $f(\phi),V(\phi)$ so that the resulting solutions are ``asymptotically Kasner''. i.e, can be matched to a Kasner solution near the initial singularity at $t=0$. 
	
	\subsubsection{Minimally coupled scalar field: Kasner spacetimes}
	\label{Sec:Kasner_Solutions}
	We first construct exact solutions (of \Eqsref{Eqs:HomoEvolve}) under the assumption that 
	\begin{align}
		V(\phi)=0,
		\quad
		f(\phi)=0.
	\end{align}
	It should be noted here that these choices are not entirely consistent with the framework described in \Sectionref{Sec:A_method_for_coupling_scalar_fields_to_gravity}. This is because, as a part of our set-up, we assumed that the potential part (of the energy momentum tensor) $V_{\mu\nu}$ is \emph{not} divergence free. If $V(\phi)=0$ then $V_{\mu\nu}=0$ which is trivially divergence free.  Nevertheless, this limit is instructive because it reveals that the dynamical coupling usually attributed to the Einstein equations is, in the Lagrangian case, partially enforced by the specific algebraic form of the energy momentum tensor. In particular, this setting is useful for understanding \emph{why} $V_{\mu\nu}$ should be divergence free. We find that our scalar field model is equivalent to the standard minimally coupled scalar model if and only if the algebraic constraint \Eqref{Eq:AlgebraicConstraint} holds. If $V(\phi)=0$ then one readily checks that 
	\begin{align}
		\alpha=1.
	\end{align}
	Using this in \Eqref{Eq:p_and_u} now gives 
	\begin{align}
		\chi(t) = \frac{\chi_\star}{t},
		\quad
		p(t)=p_\star+\chi_\star\ln(t), 
	\end{align}
	where $\chi_\star,p_\star\in\mathbb{R}$ are integration constants. The constant $\chi_\star$ can be ``absorbed'' into the $C_i$'s and hence, without loss of generality, we set $\chi_\star=1$. Similarly, $p_\star$ can be absorbed into the $\tilde{\gamma}_{i}$'s and hence we set $p_\star=0$. In addition to this, we find that  volume element $\Omega$ is 
	\begin{align}
		\Omega = \Omega_\star t^{1/3}.
	\end{align}
	where $\Omega_\star\in\mathbb{R}$. Once again we note that, without loss of generality, we can set $\Omega_\star=1$. Given all of this we find that the spatial metric $\gamma_{ab}$ can be written as 
	\begin{align}
		\gamma_{ab}=\text{diag}\left( t^{2p_1}, t^{2p_2}, t^{2p_3} \right),
	\end{align}
	where we have defined the \emph{Kasner exponents} $p_i$ as 
	\begin{align}
		p_i = C_i + \frac{1}{3}\implies p_1 + p_2 + p_3 = 1.
	\end{align}
	Turning our attention to the scalar field equations \Eqsref{Eq:nu_phi_eqs_Bianchi1} we get 
	\begin{align}
		\nu = \frac{\nu_\star}{t},
		\quad
		\phi = \phi_\star + \nu_\star\ln(t). 
	\end{align}
	Finally, \Eqref{Eq:Hamiltonian_Bianchi_1} gives
	\begin{align}
		\kappa=\frac{1}{3t^2} - \frac{1}{2t^2}\left( p_1^2 + p_2^2 + p_3^2 - \frac{1}{3} \right).
	\end{align}
	This solution is very similar to the standard Kasner scalar field solution. However, it is different in one key aspect: The solution, as presented here, \emph{does not} impose a condition on the square of the Kasner exponents. Such a condition typically arises as a consequence of the Hamiltonian constraint and can be restored within our approach here by imposing \Eqref{Eq:AlgebraicConstraint}. We find that \Eqref{Eq:AlgebraicConstraint} holds if and only if 
	\begin{align}
		p_1^2 + p_2^2 + p_3^2 = 1 - \nu_\star^2.
	\end{align}
	This condition is therefore a consequence of the algebraic constraint \Eqref{Eq:AlgebraicConstraint}. 
	
	Finally, let us comment on the requirement that the potential part $V_{\mu\nu}$ is divergence free. This condition ensures that the geometry is directly coupled to the scalar field's dynamics. If this condition is not imposed then it is possible to construct metrics whose dynamics is unaffected by scalar field perturbations. This can be seen directly from the above discussion, where we found that $\phi$ was coupled to the Kasner exponents only if one also imposes the additional condition \Eqref{Eq:AlgebraicConstraint}.
	
	\subsubsection{Near-minimal scalar field: Existence of Kasner-type spacetimes}
	\label{Sec:NearMinimal_Kasner_Solutions}
	We now discuss solutions for which $V(\phi),f(\phi)\neq 0$. In particular we investigate the question \emph{under what conditions can solutions be asymptotically matched to a Kasner scalar field solution}? To address this question it is useful to first introduce a notion of ``asymptotically Kasner'' solutions.
	\begin{defn}
		\label{Def:Kasner}
		Consider the spatially homogeneous fields $(\Omega,\chi,p,\nu,\phi)$ which are solutions of \Eqsref{Eqs:HomoEvolve} corresponding to some choice of $V(\phi),f(\phi)$ on the interval $I=\left(0,T\right]$ for some $T>0$. If there are constants $\nu_\star,\chi_\star,\Omega_{\star},p_\star,\phi_\star\in\mathbb{R}$ and functions $\tilde{\nu},\tilde{\chi},\tilde{\Omega},\tilde{p},\tilde{\phi}:\bar{I}\rightarrow\mathbb{R}$ such that 
		\begin{subequations}
			\begin{align}
				\nu(t) = (\nu_\star  + \tilde{\nu}(t))t^{-1},
				\\
				\chi(t) = (\chi_\star + \tilde{\chi}(t))t^{-1},
				\\
				\Omega(t) = (\Omega_{\star} + \tilde{\Omega}(t))t^{1/3},
				\\
				p(t) = p_\star + \chi_\star\ln(t) + \tilde{p}(t)
				\\
				\phi(t) = \phi_\star + \nu_\star\ln(t) + \tilde{\phi}(t),		
			\end{align}
			with 
			\begin{align}
				\lim_{t\rightarrow0^+}\tilde{\nu}(t)=\lim_{t\rightarrow0^+}\tilde{\chi}(t)=\lim_{t\rightarrow0^+}\tilde{\Omega}(t)=\lim_{t\rightarrow0^+}\tilde{p}(t)=\lim_{t\rightarrow0^+}\tilde{\phi}(t)=0,
			\end{align}
			then, we say that this solution is ``asymptotically Kasner''.
		\end{subequations}
	\end{defn}
	The goal now is to establish under what conditions solutions (of the Bianchi I equations, \Eqsref{Eqs:HomoEvolve}) are ``asymptotically Kasner''. We use Fuchsian analysis to find and prove these conditions. A discussion of Fuchsian methods for ODEs can be found in Appendix~\ref{Append:Fuchsian}.	In order to apply Fuchsian methods we must first re-write \Eqsref{Eqs:HomoEvolve} in terms of decaying variables. For this we define
	\begin{subequations}
		\label{Eq:decaying_vars}
		\begin{align}
			P(t) = t p(t),
			\quad
			\sigma(t) = t\nu(t),
			\quad
			x(t)=t\chi(t),
			\quad
			\omega(t) = t^{-1/3}\Omega(t),
			\quad
			\varphi(t)=t\phi(t).
		\end{align}
		Moreover, we define
		\begin{align}
			v(t) = t^{2-\epsilon}V(\phi),
			\quad
			F(t) = t^{2-\mu}f(\phi),
			\label{Eq:v_F_def}
		\end{align}
		for some constants $\epsilon,\mu>0$.  In what follows we shall treat $v(t)$ and $F(t)$ as unknown functions of time each of which are determined via an evolution equation. Such an approach is self-consistent only if \Eqref{Eq:v_F_def} is considered as an algebraic constraint on the initial data.  Given all of this we find that the lapse $\alpha$ can now be written as 
		\begin{align}
			\alpha(t,v)=\frac{1}{1-3t^\epsilon v(t)}.
		\end{align}
		From this formula it is clear that a singularity occurs if there is a $t_\star\in\left(0,T\right]$ such that $v(t_\star)=1/(3t_\star^\epsilon)$. This is a coordinate singularity and is a consequence of our particular gauge choice (see discussions in \cite{Ritchie2022} for more details). In order to avoid this gauge breakdown we restrict ourselves to time intervals $I=[0,T]$ on which $v(t)$ is ``sufficiently small''. More precisely, let $B_{R}(\mathbb{R})$ be a ball with radius $R$ centred around $0$. Then for any interval $[0,T]$ we choose $R>0$ such that $\alpha:[0,T]\times B_{R}(\mathbb{R})\rightarrow\mathbb{R}^+$ is well defined and finite for all $t\in[0,T]$.    
	\end{subequations}
	
	We find that the resulting evolution equations (for the ``decaying variables'', \Eqsref{Eq:decaying_vars}) can be written as
	\begin{subequations}
		\label{Eq:BianchSystem}
		\begin{align}
			\partial_{t}\omega(t) = t^{-1+\epsilon}\alpha(t,v)v(t)\omega(t),
			\\
			\partial_{t}x(t) = -3t^{-1+\epsilon}\alpha(t,v)v(t)x(t),
			\\
			\partial_{t}P(t)=\frac{1}{t}P(t) + x(t) + 3t^{\epsilon}v(t)\alpha(t,v)x(t),
			\\
			\partial_{t}\varphi(t) = \frac{1}{t}\varphi(t) + \sigma(t) + 3t^{\epsilon}v(t)\alpha(t,v)\sigma(t),
			\\
			\partial_{t}\sigma(t) = -3t^{-1+\epsilon}\alpha(t,v)v(t)\sigma(t) - t^{-1+\mu}\alpha(t,v) F(t),
			\\
			\partial_{t}v(t)=\frac{2-\epsilon+\mathcal{V}(t^{-1}\varphi)\sigma(t)}{t}v(t) + {3t^{-1+\epsilon}\alpha(t,v)\mathcal{V}(t^{-1}\varphi)\sigma(t)}v(t)^2,
			\label{Eq:v_evol}
			\\
			\partial_{t}F(t)=\frac{2-\mu+\mathcal{F}(t^{-1}\varphi)\sigma(t)}{t}F(t) + {3t^{-1+\epsilon}\alpha(t,v)\mathcal{F}(t^{-1}\varphi)\sigma(t)}v(t)F(t),
			\label{Eq:F_evol}
		\end{align}
		where we have defined 
		\begin{align}
			\mathcal{V}(\phi) = \frac{V^\prime(\phi)}{V(\phi)},
			\quad
			\mathcal{F}(\phi) = \frac{f^\prime(\phi)}{f(\phi)}.
		\end{align}
	\end{subequations}
	It is in terms of the system \Eqsref{Eq:BianchSystem} that we shall study existence and stability of asymptotically Kasner solutions. Note that \Eqsref{Eq:v_evol} and \eqref{Eq:F_evol} ensure that the constraints \Eqref{Eq:v_F_def} are satisfied provided they hold at the initial time $t=T$. 
	
	We now give an existence result which establishes under what conditions \Eqsref{Eq:BianchSystem} permits asymptotically Kasner solutions. 
	\begin{Thm}
		\label{Thm:Existance}
		Consider the initial value problem \Eqsref{Eq:BianchSystem} defined on the interval $I=[0,T]$, where $T>0$ has been chosen such that there is some $R>0$ so that $\alpha:I\times B_{R}(\mathbb{R})\rightarrow\mathbb{R}^+$ is well-defined and finite for all $t\in I$, and suppose that the functions $f,V:(-\infty,\infty)\rightarrow\mathbb{R}$ have been chosen such that 
		\begin{subequations}
			\label{Eq:Sols}
			\begin{align}
				\left| \frac{f^{\prime}(\phi)}{f(\phi)} \right|\le \mathcal{C}_{f},
				\quad
				\left| \frac{V^{\prime}(\phi)}{V(\phi)} \right|\le \mathcal{C}_{V},
			\end{align}
			for some $\mathcal{C}_f,\mathcal{C}_V\in\mathbb{R}^+$. If there are constants $s,r,a,b,c,d,e>0$, and $\sigma_\star\in\mathbb{R}$ such that the inequalities
			\begin{align}
				s > \min\{ a, c, d, 2 + \mathcal{C}_{V}|\sigma_\star| \}-\epsilon,
				\quad
				r > 2 - \mu + \mathcal{C}_{f}|\sigma_\star|,
				\quad
				a>c,
				\quad
				b>e,
				\label{Eq:Existance_ineqaulities}
			\end{align}
		\end{subequations}
		are satisfied then there are constants $x_\star,\omega_{\star},\varphi_{\star},P_\star\in\mathbb{R}$ and unique globally continuously differentiable functions $\tilde{v}(t),\tilde{F}(t),\tilde{x}(t),\tilde{\sigma}(t),\tilde{\omega}(t),\tilde{\varphi}(t),\tilde{P}(t):I\rightarrow\mathbb{R}$ such that
		\begin{subequations}
			\label{Eq:Sol_Expansion}
			\begin{align}
				v(t) = t^{s}\tilde{v}(t),
				\\
				F(t) = t^{r}\tilde{F}(t),
				\\
				x(t) = x_{\star} + t^a\tilde{x}(t),
				\\
				\sigma(t) = \sigma_\star + t^{b}\tilde{\sigma}(t),
				\\
				\omega(t) = \omega_{\star} + t^{d}\tilde{\omega}(t),
				\\
				\varphi(t) = t\varphi_\star + \sigma_\star t\ln(t) + t^{1+e}\tilde{\varphi}(t),
				\\
				P(t) = t P_\star + x_\star t\ln(t) + t^{1+c}\tilde{P}(t),
			\end{align}
			are solutions of \Eqsref{Eq:BianchSystem} with the property that $\tilde{v}(t),\tilde{F}(t),\tilde{x}(t),\tilde{\sigma}(t),\tilde{\omega}(t),\tilde{\varphi}(t),\tilde{P}(t)\rightarrow0$ as $t\rightarrow0$.
		\end{subequations}
	\end{Thm}
	\begin{proof}
		To prove this statement by applying the forwards Fuchsian Theorem (see Theorem~\ref{Thm:FwdsFuchs} in Appendix~\ref{Append:Fuchsian}). For this, we first plug \Eqsref{Eq:Sols} into \Eqsref{Eq:BianchSystem} in order to derive evolution equations for  $\tilde{v}(t),\tilde{F}(t),\tilde{x}(t),\tilde{\sigma}(t),\tilde{\omega}(t),\tilde{\varphi}(t),\tilde{P}(t)$. We then write the equations in Fuchsian form (see \Defref{Def:FuchsianODE} in Appendix~\ref{Append:Fuchsian}). To that end, we pick some $m$ such that
		\begin{subequations}
			\label{Eq:Bianchi_1_Existance_Fwds}
			\begin{align}
				0<m < \min\{ a-c, b - e, \epsilon+s-a,\epsilon+s-c,\epsilon+s-d \}.
			\end{align}
			The fact that such an $m$ exists is a direct consequence of \Eqref{Eq:Existance_ineqaulities}. The resulting evolution equations (for $\tilde{v}(t),\tilde{F}(t),\tilde{x}(t),\tilde{\sigma}(t),\tilde{\omega}(t),\tilde{\varphi}(t),\tilde{P}(t)$) can be written as 
			\begin{align}
				\partial_{t}u(t)= -\frac{1}{t}B(t,u)u(t) + t^{-1+m}H(t),
			\end{align}
			where $u(t)=(\tilde{\omega}(t),\tilde{x}(t),\tilde{\sigma}(t),\tilde{\varphi}(t),\tilde{P}(t),\tilde{v}(t),\tilde{F}(t))^T$ and
			\begin{align}
				B(t,u)= \text{diag}(d,a,b,e,c,s+\epsilon-2-\tilde{\mathcal{V}}(\tilde{\varphi})\sigma_{\star},r+\mu-2-\tilde{\mathcal{F}}(\tilde{\varphi})\sigma_{\star}),
				\\
				H(t,u)=(H_{1}(t,u),H_{2}(t,u),H_{3}(t,u),H_{4}(t,u),H_{5}(t,u),H_{6}(t,u),H_{7}(t,u))^T,
			\end{align}
			with 
			\begin{align}
				H_{1}(t,u) = t^{\epsilon+s-d-m}\alpha(t,t^s\tilde{v})(\omega_\star + t^{d}\tilde{\omega}(t))\tilde{v}(t),
				\\
				H_{2}(t,u) =  t^{s+\epsilon-a-m}\alpha(t,t^s\tilde{v})(x_\star+t^a\tilde{x}(t))\tilde{v}(t),
				\\
				H_{3}(t,u) =  - t^{\epsilon-m}\alpha(t,t^s\tilde{v})( t^s\tilde{v}(t)(\sigma_{\star}+t^b\tilde{\sigma}(t)) + t^{r}\tilde{F}(t) ),
				\\
				H_{4}(t,u)   =t^{b-e-m}\tilde{\sigma}(t)+3t^{s+\epsilon-e-m}\tilde{v}(t)\alpha(t,t^{s}\tilde{v}) ( \sigma_\star + \tilde{\sigma}(t) ),
				\\
				H_{5}(t,u) = t^{a-c-m}\tilde{x}(t)+ t^{\epsilon+s-c-m}\alpha(t,t^s\tilde{v}(t))(x_\star+t^a\tilde{x}(t))\tilde{v}(t),
				\\
				H_{6}(t,u) = t^{b-m}\tilde{\mathcal{V}}(\tilde{\varphi})\tilde{\sigma}(t)\tilde{v}(t) +3t^{\epsilon+s-m}\tilde{\mathcal{V}}(\tilde{\varphi})\alpha(t,t^s\tilde{v})( \sigma_\star + t^{b}\tilde{\sigma}(t) )\tilde{v}(t)^2,
				\\
				H_{7}(t,u) = t^{b-m}\tilde{\mathcal{F}}(\tilde{\varphi})\tilde{\sigma}(t)\tilde{F}(t) + 3t^{\epsilon+s-m}\tilde{\mathcal{F}}(\tilde\varphi)\alpha(t,t^s\tilde{v})( \sigma_\star + t^{b}\tilde{\sigma}(t) )\tilde{v}(t)\tilde{F}(t),
			\end{align}
			and where, for the sake of readability, we have set 
			\begin{align}
				\tilde{\mathcal{V}}(\tilde\varphi)=\mathcal{V}(\varphi_\star + \sigma_\star\ln(t) + t^{e}\tilde{\varphi}(t)),
				\quad
				\tilde{\mathcal{F}}(\tilde\varphi)=\mathcal{F}(\varphi_\star + \sigma_\star\ln(t) + t^{e}\tilde{\varphi}(t)).
			\end{align}
		\end{subequations}
		We now show that \Eqsref{Eq:Bianchi_1_Existance_Fwds} are indeed of Fuchsian form. To this end we note that, for each $t\in[0,T]$ we have that $H(t,u)$ is a continuous map which depends smoothly on $u$ and has the property that $H(t,0)=0$. Note that $H(t,u)$ depending smoothly only $u$ follows from the fact that $v\in B_{R}(\mathbb{R})$. 
		
		Now, let $(\cdot,\cdot)$ be the standard Euclidean inner product defined over $\mathbb{R}^7$. Then, in order to apply the forwards Fuchsian Theorem, we must show that there is a constant $\gamma_{1}\ge 0$ such that $(u,Bu)\ge \gamma_{1}(u,u)$. For any $u\in\mathbb{R}^7$ we have
		\begin{align}
			\begin{split}
				(u,B(t,u)u)=&du_{1}^2+au_{2}^2 + bu_{3}^2 + eu_{4}^2 + cu_{5}^2 + (s+\epsilon-2-\tilde{\mathcal{V}}(\tilde{\varphi})\sigma_{\star})u_{6}^2 
				\\
				&+ (r+\mu-2-\tilde{\mathcal{F}}(\tilde{\varphi})\sigma_{\star})u_{7}^2
				\\
				\ge& \min\{ a,b,c,d,e \}( u_1^2 + u_2^2 + u_3^2 + u_4^2 + u_5^2 ) 
				\\
				&+ (s+\epsilon-2-|\tilde{\mathcal{V}}(\tilde{\varphi})||\sigma_{\star}|)u_{6}^2 + (r+\mu-2-|\tilde{\mathcal{F}}(\tilde{\varphi})||\sigma_{\star}|)u_{7}^2
				\\
				\ge& \underbrace{\min\{ a,b,c,d,e,s+\epsilon-2-\mathcal{C}_{V}|\sigma_{\star}|,r+\mu-2-\mathcal{C}_{F}|\sigma_{\star}| \}}_{:=\gamma_{1}}(u,u).
			\end{split} 
		\end{align}
		The fact that this particular $\gamma_{1}$ is non-negative is a consequence of the assumptions \Eqref{Eq:Existance_ineqaulities}. We therefore conclude that there is a ${T}>0$ such that the initial value problem of \Eqsref{Eq:Bianchi_1_Existance_Fwds} has a unique global continuously differentiable solution $u:I\rightarrow\mathbb{R}^7$ with the property that $u\rightarrow0$ as $t\rightarrow0$. This proves the statement.
	\end{proof}
	Observe carefully that the solutions given by \Eqref{Eq:Sols} are asymptotically Kasner in the sense of \Defref{Def:Kasner}. Moreover, we note that the inequalities in \Eqref{Eq:Existance_ineqaulities} can be understood as constraints on the ``constitutive response'' of the scalar medium. Specifically, they require that the forcing term $f(\phi)$ and the potential $V(\phi)$ do not grow faster than the geometric expansion rate near the singularity. This ensures that while our model deviates from the standard variational approach, it remains within the same class of Kasner-like cosmologies near the Big Bang.

	\subsubsection{Near-minimal scalar field: Stability of Kasner-type spacetimes}
	While Theorem~\ref{Thm:Existance} tells us that asymptotically Kasner solutions exist it does not give us any information about the \emph{stability} of these solutions. By this we mean the following: Let $\{\mathring{P},\mathring{\sigma},\mathring{x},\mathring{\omega},\mathring{\varphi},\mathring{v},\mathring{F}\}$ be known asymptotically Kasner solution of \Eqsref{Eq:BianchSystem} on the interval $\left(0,T\right]$. Then, for each $(u,\mathring{u})\in\{ (P,\mathring{P}),(\sigma,\mathring{\sigma}),(x,\mathring{x}),(\omega,\mathring{\omega}),(\varphi,\mathring{\varphi}),(v,\mathring{v}),(F,\mathring{F})\}$ we solve \Eqsref{Eq:BianchSystem} on the interval $\left(0,T\right]$ with initial data 
	\begin{subequations}
		\begin{align}
			u(T)=\mathring{u}(T) + \hat{u}(T),
		\end{align}
		where $\hat{u}\in\{\hat{P},\hat{\sigma},\hat{x},\hat{\omega},\hat{\varphi},\hat{v},\hat{F}\}$ is the corresponding ``perturbed variable''. Stability now boils down to two questions: (1) If $\hat{u}(T)$ is small initially, does the full solution $\hat{u}(t)$ somehow remain ``close'' to the background solution $\mathring{u}(t)$? (2) Are these perturbed solutions still asymptotically Kasner in the sense of \Defref{Def:Kasner}?
		
		The purpose of this subsection here is to demonstrate that, under appropriate conditions, asymptotically Kasner solutions (as given by Theorem~\ref{Thm:Existance}) are indeed stable to small perturbations. Here, we focus exclusively on the Bianchi I setting.	It would of course be interesting to investigate stability in the context of spatially \emph{inhomogeneous} perturbations. However, proving such a thing is complicated and goes well beyond the scope of our focus here. 
		
		Now, for each $(u,\mathring{u})\in\{ (P,\mathring{P}),(\sigma,\mathring{\sigma}),(x,\mathring{x}),(\omega,\mathring{\omega}),(\varphi,\mathring{\varphi}),(v,\mathring{v}),(F,\mathring{F})\}$ we define the corresponding perturbed variable $\hat{u}\in\{\hat{P},\hat{\sigma},\hat{x},\hat{\omega},\hat{\varphi},\hat{v},\hat{F}\}$ as 
		\begin{align}
			\hat{u}(t) = u(t) - \mathring{u}(t).
		\end{align}
	\end{subequations}
	The resulting evolution equations for the fields 
	$\{\hat{P},\hat{\sigma},\hat{x},\hat{\omega},\hat{\varphi},\hat{v},\hat{F}\}$ are
	\begin{subequations}
		\label{Eq:PertrubedEqs}
		\begin{align}
			\partial_{t}\hat{\omega} = t^{-1+\epsilon}\alpha(t,\hat{v}+\mathring{v})\left( (\mathring{v}(t)+\hat{v}(t))\hat{\omega}(t) + \alpha(t,\mathring{v})\mathring{\omega}(t)\hat{v}(t) \right)
			\\
			\partial_{t}\hat{x}(t) = -3t^{-1+\epsilon}\alpha(t,\hat{v}+\mathring{v})( (\hat{v}(t) + \mathring{v}(t) )\hat{x}(t) + \alpha(t,\mathring{v})\mathring{x}(t)\hat{v}(t) ),
			\\
			\partial_{t}\hat{P}(t) = \frac{1}{t}\hat{P}(t) + \hat{x}(t) +3t^{\epsilon}\alpha(t,\mathring{v}+\hat{v})( (\mathring{v}(t)+\hat{v}(t))\hat{x}(t) + \alpha(t,\mathring{v})\mathring{x}(t)\hat{v}(t) ),
			\\
			\partial_{t}\hat{\varphi}(t) = \frac{1}{t}\hat{\varphi}(t) + \hat{\sigma}(t) +3t^{\epsilon}\alpha(t,\mathring{v}+\hat{v})( (\mathring{v}(t)+\hat{v}(t))\hat{\sigma}(t) + \alpha(t,\mathring{v})\mathring{\sigma}(t)\hat{v}(t) ),
			\\
			\partial_{t}\hat{\sigma}(t)= -3t^{-1+\epsilon}\alpha(t,\hat{v}+\mathring{v})( (\hat{v}(t)+\mathring{v}(t))\hat{\sigma}(t) + \alpha(t,\mathring{v})\hat{v}(t)\mathring\sigma(t))	- t^{-1+\mu}h_{\sigma},
			\\
			\partial_{t}\hat{v}(t)= \frac{2-\epsilon + \mathcal{V}(t^{-1}(\hat{\varphi}+\mathring{\varphi}))(\mathring{\sigma}(t)+\hat{\sigma}(t))}{t}\hat{v}(t) + \frac{h_{v}}{t}\mathring{v}(t) + t^{-1+\epsilon}\alpha(t,\mathring{v}+\hat{v})\hat{h}_{v},
			\\
			\partial_{t}\hat{F}(t) = \frac{2-\mu+\mathcal{F}(t^{-1}(\mathring{\varphi}+\hat{\varphi}))(\mathring{\sigma}(t)+\hat{\sigma}(t))}{t}\hat{F}(t) + \frac{h_{f}}{t}\mathring{F}(t) + t^{-1+\epsilon}\alpha(t,\mathring{v}+\hat{v})\hat{h}_{f},
		\end{align}
		where, for the sake of readability, we have defined
		\begin{align}
			h_{\sigma} = \alpha(t,\hat{v}+\mathring{v})(\hat{F}(t) + 3t^{\epsilon}\alpha(t,\mathring{v})\mathring{F}(t)\hat{v}(t))&,
			\\
			h_{v} = \mathcal{V}(t^{-1}(\hat{\varphi}+\mathring{\varphi}))\hat{\sigma}(t) + (\mathcal{V}(t^{-1}(\hat{\varphi}+\mathring{\varphi}))-\mathcal{V}(t^{-1}\mathring{\varphi}))\mathring{\sigma}(t)&,
			\\
			h_{f} = \mathcal{F}(t^{-1}(\mathring{\varphi}+\hat\varphi))\hat{\sigma}(t)+(\mathcal{F}(t^{-1}(\mathring{\varphi}+\hat\varphi))-\mathcal{F}(t^{-1}\mathring{\varphi})\mathring{\sigma}(t)&,
			\\
			\hat{h}_{v} = 3( (\hat{v}(t)+\mathring{v}(t))^2\hat{\sigma}(t) + (\hat{v}(t)+2\mathring{v}(t))\hat{v}(t)\mathring{\sigma}(t) + 3t^{\epsilon}\alpha(t,\mathring{v})\mathring{v}(t)^2\mathring{\sigma}(t)\hat{v}(t) )&, 
			\\
			\begin{split}
				\hat{h}_{f} = 3(\mathring{F}(t)\mathring{v}(t)\hat{\sigma}(t)+\mathring{F}(t)\hat{v}(t)(\mathring{\sigma}(t)+\hat{\sigma}(t))+(\mathring{v}(t)+\hat{v}(t))(\mathring{\sigma}(t)+\hat{\sigma}(t))\hat{F}(t))&
				\\
				+9t^{\epsilon}\alpha(t,\mathring{v})\mathring{\sigma}(t)\mathring{v}(t)^2\hat{v}(t)&.
			\end{split}
		\end{align}
	\end{subequations}		
	We now show that if the perturbed variables are sufficiently small at the initial time $t=T$ then they shall remain small for all $t\in[0,T]$. The question of whether or not the resulting solutions are asymptotically Kasner shall be addressed after.
	\begin{Thm}
		\label{Thm:Bounded}
		Consider \Eqsref{Eq:PertrubedEqs} and suppose that the fields $\{\mathring{P},\mathring{\sigma},\mathring{x},\mathring{\omega},\mathring{\varphi},\mathring{v},\mathring{F}\}$ are given as in \Eqsref{Eq:Sol_Expansion}. If the constants $\epsilon,\mu,\mathcal{C}_V,\mathcal{C}_F$ and $\sigma_{\star}$ satisfy the inequalities
		\begin{subequations}
			\begin{align}
				2-\epsilon- \mathcal{C}_{V}|\sigma_\star|>0,
				\quad
				2-\mu- \mathcal{C}_{F}|\sigma_\star|>0,
				\label{Eqs:StabilityInequalities}
			\end{align}
			then there are $R,\delta>0$ such that the initial value problem of \Eqsref{Eq:PertrubedEqs}, on the interval $[0,T]$, has a unique global continuously differentiable solution with the property 
			\begin{align}
				|\hat{P}(t)|^2 + |\hat{\sigma}(t)|^2 +|\hat{x}(t)|^2 + |\hat{\omega}(t)|^2 + |{\varphi}(t)|^2+ |\hat{v}(t)|^2 + |\hat{F}(t)|^2 < R,
			\end{align}
			for all $t\in [0,T]$, provided
			\begin{align}
				|\hat{P}(T)|^2 + |\hat{\sigma}(T)|^2 +|\hat{x}(T)|^2 + |\hat{\omega}(T)|^2 + |{\varphi}(T)|^2+ |\hat{v}(T)|^2 + |\hat{F}(T)|^2 < \delta.
			\end{align}
		\end{subequations}
	\end{Thm}
	\begin{proof}
		We prove this via straightforward application of the backwards Fuchsian Theorem (see Theorem~\ref{Thm:BcksFuchs} in Appendix~\ref{Append:Fuchsian}). Now, given the assumptions above we find that, for any $m$ such that
		\begin{align}
			0<m<\min\{1,\epsilon,\mu,b,r,s\},
		\end{align}
		\Eqsref{Eq:PertrubedEqs} can be written as 
		\begin{subequations}
			\label{Eq:uhat}
			\begin{align}
				\partial_{t}\hat{u}(t) = \frac{1}{t}\mathcal{B}(t,\hat{u})\hat{u} + t^{-1+m}\hat{H}(t,\hat{u}),
			\end{align}
			where $\hat{u}(t)=(\hat{\omega}(t),\hat{x}(t),\hat{v}(t),\hat{F}(t),\hat{\sigma}(t),\hat{P}(t),\hat{\varphi}(t))^T$ and
			\begin{align}
				\hat{H}(t,\hat{u}) =&\; (\hat{H}_{1}(t,\hat{u}),\hat{H}_{1}(t,\hat{u}),\hat{H}_{1}(t,\hat{u}),\hat{H}_{1}(t,\hat{u}),\hat{H}_{1}(t,\hat{u}),\hat{H}_{1}(t,\hat{u}),\hat{H}_{1}(t,\hat{u}))^T,
				\\
				\begin{split}
					\mathcal{B}(t,\hat{u}) =&\; \text{diag}\left( 0, 0, 2-\epsilon + \mathcal{V}(t^{-1}(\hat{\varphi}+\mathring{\varphi}))({\sigma}_\star + \hat{\sigma}(t)) \right., 
					\\
					&\quad\quad\quad \left. 2-\mu + \mathcal{F}(t^{-1}(\hat{\varphi}+\mathring{\varphi}))({\sigma}_\star +  \hat{\sigma}(t)), 0, 1, 1 \right),
				\end{split}
			\end{align}
			with
			\begin{align}
				\hat{H}_{1}(t,\hat{u}) = t^{\epsilon-m}\alpha(t,\hat{v}+\mathring{v})\left( (\mathring{v}(t)+\hat{v}(t))\hat{\omega}(t) + \alpha(t,\mathring{v})\mathring{\omega}(t)\hat{v}(t) \right),
				\\
				\hat{H}_{2}(t,\hat{u})= -3t^{\epsilon-m}\alpha(t,\hat{v}+\mathring{v})( (\hat{v}(t) + \mathring{v}(t) )\hat{x}(t) + \alpha(t,\mathring{v})\mathring{x}(t)\hat{v}(t) ),
				\\
				\hat{H}_{3}(t,\hat{u})= { \mathcal{V}(t^{-1}(\hat{\varphi}+\mathring{\varphi}))\tilde{\sigma}(t)}\hat{v}(t)t^{b-m} + {h_{v}}{t^{s-m}}\tilde{v}(t) + t^{\epsilon-m}\alpha(t,\mathring{v}+\hat{v})\hat{h}_{v},
				\\
				\hat{H}_{4}(t,\hat{u})={\mathcal{F}(t^{-1}(\mathring{\varphi}+\hat{\varphi}))\tilde{\sigma}(t)}\hat{F}(t)t^{b-m} + {h_{f}}\tilde{F}(t)t^{r-m} + t^{\epsilon-m}\alpha(t,\mathring{v}+\hat{v})\hat{h}_{f},
				\\
				\hat{H}_{5}(t,\hat{u})=-3t^{\epsilon-m}\alpha(t,\hat{v}+\mathring{v})( (\hat{v}(t)+\mathring{v}(t))\hat{\sigma}(t) + \alpha(t,\mathring{v})\hat{v}(t)\mathring\sigma(t))	- t^{\mu-m}h_{\sigma},
				\\
				\hat{H}_{6}(t,\hat{u})=\hat{x}(t)t^{1-m} +3t^{1+\epsilon-m}\alpha(t,\mathring{v}+\hat{v})( (\mathring{v}(t)+\hat{v}(t))\hat{x}(t) + \alpha(t,\mathring{v})\mathring{x}(t)\hat{v}(t) ),
				\\
				\hat{H}_{7}(t,\hat{u})= \hat{\sigma}(t)t^{1-m} +3t^{1+\epsilon-m}\alpha(t,\mathring{v}+\hat{v})( (\mathring{v}(t)+\hat{v}(t))\hat{\sigma}(t) + \alpha(t,\mathring{v})\mathring{\sigma}(t)\hat{v}(t) ),
			\end{align}
		\end{subequations}
		Observe carefully that $\hat{H}(t,\hat{u})$ is a continuous map, that for each $t\in[0,T]$, depends smoothly on $\hat{u}$. Thus, in order to apply Theorem~\ref{Thm:BcksFuchs} it only remains to show that there is some $R>0$ such that $(\hat{u},\mathcal{B}(t,\hat{u})\hat{u})\ge0$ for all $|\hat{u}|< R$. Given that $\mathcal{B}(t,\hat{u})$ is diagonal it is sufficient to show that each entry (of $\mathcal{B}(t,\hat{u})$) is non-negative. There are two non-trivial entries: $\mathcal{B}_{33}(t,\hat{u})$ and $\mathcal{B}_{44}(t,\hat{u})$. Starting with $\mathcal{B}_{33}(t,\hat{u})$ we get
		\begin{subequations}
			\begin{align}
				\begin{split}
					2-\epsilon + \mathcal{V}(t^{-1}(\hat{\varphi}+\mathring{\varphi}))({\sigma}_\star + \hat{\sigma}(t)) \ge 2-\epsilon - |\mathcal{V}(t^{-1}(\hat{\varphi}+\mathring{\varphi}))||{\sigma}_\star + \hat{\sigma}(t)|
					\\
					\ge 2-\epsilon - \mathcal{C}_{V}|{\sigma}_\star + \hat{\sigma}(t)| \ge 2-\epsilon - \mathcal{C}_{V}|{\sigma}_\star| - \mathcal{C}_{V}|\hat{\sigma}(t)|,
				\end{split}
			\end{align}
			which is positive provided 
			\begin{align}
				|\hat{\sigma}(t)|< \frac{2-\epsilon - \mathcal{C}_{V}|{\sigma}_\star|}{ \mathcal{C}_{V}}.
			\end{align}
			A similar calculation can be repeated for $\mathcal{B}_{44}(t,\hat{u})$. In doing so we find that 
			\begin{align}
				R= \min\left\{ \frac{2-\epsilon - \mathcal{C}_{V}|{\sigma}_\star|}{ \mathcal{C}_{V}}, \frac{2-\mu - \mathcal{C}_{F}|{\sigma}_\star|}{ \mathcal{C}_{F}} \right\}.
			\end{align}
			The fact that $R$ exists and is positive is a consequence of \Eqref{Eqs:StabilityInequalities}. We therefore have that $(\hat{u},\mathcal{B}(t,\hat{u})\hat{u})\ge0$ for all $|\hat{u}|< R$. It follows then that we can apply the backwards Fuchsian Theorem (Theorem~\ref{Thm:BcksFuchs}, Appendix~\ref{Append:Fuchsian}) to conclude that there is a $\delta>0$ such that the initial value problem of \Eqsref{Eq:PertrubedEqs} has a unique global continuously differentiable solution $\hat{u}:[0,T]\times B_{R}(\mathbb{R}^7)\rightarrow\mathbb{R}$ with the property 
			\begin{align}
				|\hat{u}(t)| < R,
			\end{align}
			\emph{provided} 
			\begin{align}
				|\hat{u}(T)|<\delta.
			\end{align}
			This proves the statement.
		\end{subequations}
	\end{proof}
	Observe carefully that, in order for the inequalities \Eqref{Eqs:StabilityInequalities} to hold, we must have $\epsilon,\mu<2$. This is consistent with results presented in \cite{Ritchie2022}. 
	
	Theorem~\ref{Thm:Bounded} tells us that if the perturbation is sufficiently small at $t=T$ then they remain small as $t\rightarrow0$. However, we have yet to show that the resulting solutions are asymptotically Kasner (in the sense of \Defref{Def:Kasner}). That is the purpose of the following Corollary.
	\begin{Cor}
		Consider \Eqsref{Eq:PertrubedEqs} and suppose that the assumptions of Theorem~\ref{Thm:Bounded} hold. Then, the corresponding solutions $\{p,\nu,\chi,\Omega,{\phi}\}$ are asymptotically Kasner.
	\end{Cor}
	\begin{proof}
		We prove this in two steps. In a first we show that the perturbed variables $(\hat{\omega},\hat{x},\hat{\sigma})$ satisfy the appropriate bounds. Then, in a second step, we show that the \emph{physical unknowns} $(\phi,p)$ have the appropriate structure.   
		
		We now show that $\hat{\omega}$ decays appropriately. For this we use a standard Cauchy sequence argument: For any two times $t,\tilde{t}\in\left(0,T\right]$, we have
		\begin{subequations}
			\begin{align}
				\hat{\omega}(t) - \hat{\omega}(\tilde{t}) = \int_{\tilde{t}}^{t} s^{-1+m}\hat{H}_1(s,\hat{u}(s))ds, 
			\end{align}
			which follows from \Eqref{Eq:uhat}. Then,
			\begin{align}
				|\hat{\omega}(t) - \hat{\omega}(\tilde{t})|\le \int_{\tilde{t}}^{t} s^{-1+m}|\hat{H}_1(s,\hat{u}(s))|ds \le \frac{C_\omega}{m}(t^m-\tilde{t}^m),
				\label{Eq:omega_CauchyEstimate}
			\end{align}
			where we have defined 
			\begin{align}
				C_\omega = \sup_{{(t,\hat{u})\in[0,T]\times B_{R}(\mathbb{R}^7)}}|\hat{H}_{1}(t,\hat{u})|.
			\end{align}
		\end{subequations}
		Now, given any monotonic sequence of times $(t_n)$ in $\left(0,T\right]$ with the property that $(t_n)\rightarrow 0$ as $n\rightarrow\infty$ we have that the sequence $(t_n^m)$ also converges to zero. As a consequence, we conclude that $(t_n^m)$ is a \emph{Cauchy sequence}. i.e., for every $\delta>0$ there exists an integer $N>0$ such that 
		\begin{subequations}
			\begin{align}
				|t_{n_1}^m-t_{n_2}^m|<\delta,
				\label{Eq:2ndODE_tCauchy}
			\end{align}
			for all $n_1,n_2\ge N$. If $n_1\le n_2$ then \Eqref{Eq:omega_CauchyEstimate} gives 
			\begin{align}
				|\hat{\omega}(t_{n_1})-\hat{\omega}(t_{n_2})|\le \frac{C_\omega}{m}|t_{n_1}^m-t_{n_2}^m|<\frac{C_\omega}{m}\delta.
			\end{align}
			Thus $(\hat{\omega}(t_n))$ is a Cauchy sequence and hence converges to some limit which we label as
			\begin{align}
				\hat{\omega}_\star = \lim_{n\rightarrow\infty}\hat{\omega}(t_n). 
			\end{align}
			Given that this limit exists we can return to \Eqref{Eq:omega_CauchyEstimate} and take the limit $\tilde{t}\rightarrow0$ to find
			\begin{align}
				|\hat{\omega}(t) - \hat{\omega}_\star|\le \frac{C_\omega}{m}t^m.
			\end{align}
		\end{subequations}
		Before continuing we must first check that the limit $\hat{\omega}_\star$ is independent of the chosen sequence $(t_n)$. To this end, let $(\tau_n)$ be another monotonic sequence of times in $\left(0,T\right]$ with the property that $(\tau_n)\rightarrow 0$ as $n\rightarrow\infty$. Proceeding as before we conclude that $(\hat{\omega}(\tau_n))$ converges to some limit $\bar{\omega}_\star$. Considering now the sequence $W_{n}=\hat{\omega}(t_n)-\hat{\omega}(\tau_n)$ we find
		\begin{subequations}
			\begin{align}
				\begin{split}
					|W_{n_1}-W_{n_2}|&=|\hat{\omega}(t_{n_1})-\hat{\omega}(t_{n_2})-(\hat{\omega}(\tau_{n_1})-\hat{\omega}(\tau_{n_2}))|
					\\
					&\le |\hat{\omega}(t_{n_1})-\hat{\omega}(t_{n_2})|+|\hat{\omega}(\tau_{n_1})-\hat{\omega}(\tau_{n_2})|<\frac{2C_\omega}{m}\delta.
				\end{split}
			\end{align}
			and hence $W_n$ is also a Cauchy sequence. There therefore exists a limit $W_\star$ defined as 
			\begin{align}
				W_\star = \lim_{n\rightarrow\infty}W_{n}=\hat{\omega}_\star-\bar{\omega}_\star.
			\end{align}
			Returning again to \Eqref{Eq:omega_CauchyEstimate} we pick $t\in(t_n)$ and $\tilde{t}\in (\tau_n)$ so that \Eqref{Eq:omega_CauchyEstimate} can be written as 
			\begin{align}
				|W_n|=|\hat{\omega}(t_n)-\hat{\omega}(\tau_n)|\le \frac{C_\omega}{m}|t_n^m-\tau_n^m|.
			\end{align}
			Taking the limit $n\rightarrow\infty$ and using that $t_n,\tau_n\rightarrow0$ we get 
			\begin{align}
				|W_\star|=0\implies \hat{\omega}_\star=\bar{\omega}_\star,
			\end{align}
		\end{subequations}
		and hence the limit is independent of the chosen sequence. This argument is standard and can be repeated for the variables $\hat{x},\hat{\sigma}$ to give
		\begin{subequations}
			\begin{align}
				|\hat{x}(t) - \hat{x}_\star|\le \frac{C_x}{m}t^m,
				\quad
				|\hat{\sigma}(t) - \hat{\sigma}_\star|\le \frac{C_\sigma}{m}t^m,
			\end{align}
			where 
			\begin{align}
				C_x = \sup_{{(t,\hat{u})\in[0,T]\times B_{R}(\mathbb{R}^7)}}|\hat{H}_2(t,\hat{u})|,
				\quad
				C_\sigma = \sup_{{(t,\hat{u})\in[0,T]\times B_{R}(\mathbb{R}^7)}}|\hat{H}_5(t,\hat{u})|
			\end{align}
		\end{subequations}
		It follows that for any constants $\hat{a},\hat{b},\hat{d}>0$ with the property that
		\begin{subequations}
			\begin{align}
				\hat{a}<\min\{a,m\},
				\quad
				\hat{b}<\min\{a,b\},
				\quad
				\hat{d}<\min\{d,m\},
			\end{align}
			there are functions $\check{\omega},\check{x},\check{\sigma}:[0,T]\rightarrow\mathbb{R}$ such that 
			\begin{align}
				\omega(t) = \omega_\star + \hat{\omega}_\star + t^{\hat{d}}\check{\omega}(t),
				\quad
				x = x_\star + \hat{x}_\star + t^{\hat{a}}\check{x}(t),
				\quad
				\sigma = \sigma_\star + \hat{\sigma}_\star + t^{\hat{b}}\check{\sigma}(t),
			\end{align}
		\end{subequations}
		with the property that $\check{\omega}(t),\check{x}(t),\check{\sigma}(t)\rightarrow 0$ as $t\rightarrow0$. The corresponding physical variables $(\Omega,\chi,\nu)$ (calculated via \Eqsref{Eq:decaying_vars}) therefore satisfy the conditions of \Defref{Def:Kasner}. 
		
		It now only remains to show that $\phi$ and $p$ are also asymptotically Kasner. In this case it is more convenient to work directly with the physical variables instead of the decaying ones. Starting with $\phi$ we note that, for any two times $t,\tilde{t}\in\left(0,T\right]$ we have
		\begin{subequations}
			\begin{align}
				\begin{split}
					\phi(t)-\phi(\tilde{t}) = \int_{\tilde{t}}^{t}\alpha(s,\mathring{v}+\hat{v})\nu(s)ds 
					=
					(\sigma_{\star}+\hat{\sigma}_\star)(\ln(t)-\ln(\tilde{t}))
					\\
					+ \int_{\tilde{t}}^{t}\left( s^{\epsilon-1}\alpha(s,\mathring{v}+\hat{v})(\mathring{v}(s)+\hat{v}(s))\sigma(s) + s^{-1+\hat{b}}\check{\sigma}(s) \right)ds.
				\end{split}
			\end{align}
			Rearranging and taking the absolute value gives 
			\begin{align}
				| \phi(t)-(\sigma_\star+\hat{\sigma}_{\star})\ln(t) - (\phi(\tilde{t})-(\sigma_\star+\hat{\sigma}_{\star})\ln(\tilde{t})) | \le \frac{C_{\phi}}{\hat{\epsilon}}( t^{\hat{\epsilon}}-\tilde{t}^{\hat{\epsilon}} ),
			\end{align}
			where $\hat{\epsilon}=\min\{ \epsilon, \hat{b} \}$ and
			\begin{align}
				C_\phi = \sup_{{(t,\hat{u})\in[0,T]\times B_{R}(\mathbb{R}^7)}}|t^{\epsilon-\hat{\epsilon}}\alpha(t,\mathring{v}+\hat{v})(\mathring{v}(t)+\hat{v}(t))\sigma(t) + t^{\hat{b}-\hat{\epsilon}}\check{\sigma}(t)|.
			\end{align}
			From here we can proceed in exactly the same way as above. i.e., we pick some Cauchy sequence $(t_n)$ in $\left(0,T\right]$ with the property that $(t_n)\rightarrow0$ as $n\rightarrow\infty$. Then, applying \Eqref{Eq:2ndODE_tCauchy}, we find that 
			\begin{align}
				| \phi(t_{n_1})-(\sigma_\star+\hat{\sigma}_{\star})\ln(t_{n_1}) - (\phi({t_{n_2}})-(\sigma_\star+\hat{\sigma}_{\star})\ln(t_{n_2})) |<\frac{C_\phi}{\hat{\epsilon}}\delta,
			\end{align}
			and hence $\phi(t_{n})-(\sigma_\star+\hat{\sigma}_{\star})\ln(t_{n})$ is a Cauchy sequence and has a limit which we label as 
			\begin{align}
				\phi_{\star}=\lim_{n\rightarrow\infty}|\phi(t_{n})-(\sigma_\star+\hat{\sigma}_{\star})\ln(t_{n})|.
			\end{align}
			One can apply a similar argument as before to show that this limit is independent of the chosen sequence $(t_n)$.
		\end{subequations}
		Applying the same procedure to $p(t)$ gives 
		\begin{subequations}
			\begin{align}
				|p(t) - p_{\star} - (x_\star+\hat{x}_\star)\ln(t) | \le \frac{C_p}{\check{\epsilon}}t^{\check{\epsilon}},
			\end{align}
			where $\check{\epsilon}=\min\{ \epsilon, \hat{a} \}$ and 
			\begin{align}
				C_\phi = \sup_{{(t,\hat{u})\in[0,T]\times B_{R}(\mathbb{R}^7)}}|t^{\epsilon-\check{\epsilon}}\alpha(t,\mathring{v}+\hat{v})(\mathring{v}(t)+\hat{v}(t))x(t) + t^{\hat{a}-\check{\epsilon}}\check{x}(t)|.
			\end{align}
		\end{subequations}
		We have then that for any constants $\hat{e},\hat{c}>0$ with the property that 
		\begin{subequations}
			\begin{align}
				\hat{e}<\min\{ e, \hat{\epsilon} \},
				\quad
				\hat{c}<\min\{ c, \check{\epsilon} \},
			\end{align}
			there are functions $\check{\omega},\check{p}:[0,T]\rightarrow \mathbb{R}$ such that 
			\begin{align}
				\phi(t) = \phi_{\star} + (\sigma_\star+\hat{\sigma}_\star)\ln(t) + t^{\check{e}}\check{\phi}(t),
				\quad
				p(t) = p_{\star} + (x_\star+\hat{x}_\star)\ln(t) + t^{\check{c}}\check{p}(t),
			\end{align}
			with the property that $\check{\phi}(t),\check{p}(t)\rightarrow 0$ as $t\rightarrow0$. This now proves the statement. 
		\end{subequations}
	\end{proof}
	
	\section{Conclusion}
	The primary goal of this work was to introduce a generic method for coupling matter fields to gravity without the use of a variational principle or Lagrangian. Our particular focus here was on scalar field models. To do this we ``split'' the energy momentum tensor into two pieces: the potential and kinetic parts. The potential part was considered as freely specifiable, and describes how the matter content interacts with gravity. The kinetic part is then determined as a solution of the Bianchi identity. 
	
	Our work here has focused exclusively scalar field couplings. However, the framework presented here also allows for \emph{non scalar field couplings}. The only restriction is that the ``potential part'' should have a non-vanishing divergence. 
	
	For our particular scalar field coupling, we chose the potential part to be determined by a single function. This coupling was chosen as it is consistent with standard scalar field models such as the minimal and $k$-essence scalar fields. Indeed, we were able to demonstrate that such models were contained within our framework. 
	
	In regards to the kinetic part (of the energy momentum tensor), we found that the Bianchi identities alone are underdetermined and therefore cannot be used to fully determine the kinetic part. This suggests that some piece of the kinetic part must be chosen before the equations can be solved. Within our framework, we identified this freedom as a constitutive choice, analogous to the material response functions in continuum mechanics. To explore this freedom we made use of the $(3+1)$-decomposition. Within this framework, we established under what conditions the Bianchi identities permitted a well-posed Cauchy problem. In addition we identified a field that could ``naturally'' be considered as freely specifiable. This choice (of free data) is not the only possible one and it should be emphasised that other choices are possible and may somehow be more ``preferable''.  
	
	We then made use of our formalism to investigate Bianchi I solutions of (what we have dubbed) \emph{near-minimal} scalar fields. For this model the forcing term (in the equation of motion) \emph{does not} coincide with the derivative of the potential function. Nonetheless, all other freedoms \emph{do} match the standard minimally coupled scalar field (with a non-zero potential) and as such we regard this model as being ``similar'' to a minimally coupled scalar field.
	
	In this setting, we investigated questions related to stable Big Bang formation. In particular, we established conditions under which the resulting solutions (of this near-minimal scalar field model) could be asymptotically matched to a Kasner scalar field solution near the initial singularity. In addition, we found that these solutions were stable to sufficiently small perturbations, provided they satisfied some appropriate conditions.
	
	Our initial explorations here are intriguing and suggest that there is a much larger class of permissible scalar field couplings that are consistent with General Relativity. While the immediate physical applicability of these non-variational models remains a subject for future investigation, our results demonstrate that the mathematical consistency of General Relativity does not strictly require a Lagrangian origin for the matter sector. We now briefly discuss three possible avenues for future research.
	
	The first clear direction would be to consider our framework within the context of \emph{black holes}. For standard scalar field models, the behaviour of the scalar field around a black hole is tightly constrained by traditional ``no-hair'' theorems (see, for example, \cite{herdeiro2015:scalarhair}). Our framework does not have a Lagrangian and as such it may possible to specify a highly non-linear equation of motion, for the scalar field, that could evade these theorems. In the context of the early universe, modifying the scalar field's kinetic and potential evolution without the rigid constraints of an action could alter the primordial power spectrum on small scales. This has direct relevance for the formation of primordial black holes, as well as the generation of early scalar field perturbations that could serve as seeds for supermassive black holes. A detailed analysis of black hole solutions and their stability within this non-variational framework is therefore an important direction for future research.	
	
	Second, it would be of particular interest to investigate the role of free data with regards to conservation laws. In standard variational formulations of GR, the covariant conservation of the energy-momentum tensor is fundamentally tied to the symmetries of the action via the Noether theorems. Due to the fact that our formalism does not have a Lagrangian (and does not arise from an action), traditional Noether conservation laws do not directly apply. Instead, conservation is enforced geometrically via the Bianchi identity. To frame this in terms of underlying symmetries without an action, it may be necessary reformulate the conservation laws explicitly in terms of spacetime Killing vectors, where the contraction $J^\mu = T^{\mu\nu}\xi_\nu$ with a Killing field $\xi_\nu$ yields a conserved current (i.e., $\nabla_{\mu}J^\mu = 0$).
	
	For our final avenue, of possible future work, it would be interesting to investigate different equations of motion. In the work presented here we have focused exclusively on a wave-equation as it is consistent with standard scalar field models. However, more general equations of motion are also possible. One could, for example, consider a first order transport equation. Indeed, our framework itself does not impose any particular restriction of the specific choice of the scalar fields equation of motion. Investigating other physically relevant equations would therefore be a fascinating line of inquiry that could allow for scalar field dynamics that is otherwise unobtainable with the standard Lagrangian couplings.

	\newpage
	\bibliographystyle{unsrt}
	\bibliography{bibfile}
	
	\newpage
	\begin{appendices}
		\section{Fuchsian theory for ODEs}
		\label{Append:Fuchsian}
		The purpose of this appendix to discuss Fuchsian techniques for ODEs. To help motivate the discussion here we first consider an example. Suppose that $u(t)$ is some unknown whose behaviour is governed by the ODE
		\begin{subequations}
			\begin{align}
				\partial_{t}u(t)=\frac{a}{t}u(t),
				\quad
				u(T_0) = u_0.
			\end{align}
			The goal is to solve this equation on the interval $[0,T_0]$. However, due to the presence of a $1/t$-term the evolution equation itself is formally singular. Nevertheless, one can integrate to find
			\begin{align}
				u(t)= t^{a}T_0^{-a}u_0.
			\end{align} 
		\end{subequations}
		Here we see that the function $u(t)$ is well-defined in the limit $t\rightarrow0$ provided $a>0$. In particular, $u(t)$ is well-defined and finite near $t=0$ even though the ODE itself is \emph{not}.
		
		This is the prototypical example of \emph{Fuchsian} equation, and it turns out that there is a much larger class of ODEs which are formally singular at $t=0$ whose solutions are nonetheless well-defined for all $t\in[0,T_0]$. This class is defined as follows:
		\begin{defn}
			\label{Def:FuchsianODE}
			Given constants $T_{0}>0,p>0,$ and $R>0$. A system of ordinary differential equations 
			\begin{align}
				\partial_{t}u(t)=\frac{1}{t}\mathcal{B}(t,u(t))u(t) + \tilde{F}(t) + t^{p-1}F_{0}(t,u(t)),
				\label{Eq:FuchsODE}
			\end{align}
			for an $\mathbb{R}^N$-valued unknown $u(t)$ is called a ``Fuchsian ODE system'' defined on the time interval $I=\left(0,T_0 \right]$, provided $\mathcal{B}:I\times B_{R}(\mathbb{R}^N) \rightarrow\mathbb{R}^{N\times N}$ is a continuous map $\mathcal{B}(t,v)$ that, for each $t\in I$, depends smoothly on the variable $v$, and, $\tilde{F}:I\rightarrow\mathbb{R}^N$ is continuous with 
			\begin{align}
				\int_{0}^{T_0}|\tilde{F}(s)|ds<\infty,
			\end{align}
			and, finally, $F_0:\bar{I}\times B_{R}(\mathbb{R}^N)\rightarrow\mathbb{R}^N$ is a continuous map $F_{0}(t,v)$ which, for each $t\in \bar{I}=[0,T_0]$, depends smoothly on $v$ and has the property that 
			\begin{align}
				F_{0}(t,0)=0,
				\quad 
				t\in  \bar{I}.
				\label{Eq:F=0}
			\end{align} 
		\end{defn}
		We utilize this definition in the following way: Suppose one aims to solve an equation of the form $F[\phi,t]=0$. Then the goal is to find functions $\phi_{\star}(t)$ and $f(t)$ such that $F[\phi_\star(t)+f(t)u(t),t]=0$ is a Fuchsian ODE for the new unknown $u(t)$. Given such a transformation one then applies the following theorem to establish existence. 
		\begin{Thm}[Forwards Fuchsian Theorem]
			\label{Thm:FwdsFuchs}
			Consider constants $T_{0}>0,p>0,$ and $R>0$, and a Fuchsian ODE system 
			\begin{align}
				\partial_{t}u(t)=\frac{1}{t}\mathcal{B}(t,u(t))u(t) + \tilde{F}(t) + t^{p-1}F_{0}(t,u(t)), 
				\label{Eq:FuchsianIVP_Fwds}
			\end{align}
			in the sense of \Defref{Def:FuchsianODE}. Suppose in addition that there exists a $\gamma_{1}\ge 0$ such that
			\begin{align}
				(v,\mathcal{B}(t,v) v) \le -\gamma_{1}|v|^2,
				\label{Eq:FuchsianODE_Fwd_y1}
			\end{align} 
			for all $v\in B_{R}(\mathbb{R}^N)$ and $t\in I$.
			
			Then there exists a $\delta>0$ such that the initial value problem \Eqref{Eq:FuchsianIVP_Fwds} has a unique global continuously differentiable solution $u:I\rightarrow \mathbb{R}^N$ such that 
			\begin{align}
				u(t)\rightarrow 0 \quad \text{as} \quad t\rightarrow 0
			\end{align}
			provided 
			\begin{align}
				\int_{0}^{T_0}|\tilde{F}(s)|ds\le\delta,
			\end{align} 
			and there is a constant $\lambda>0$ such that
			\begin{align}
				\int_{0}^{T_0}s^{-\lambda}|\tilde{F}(s)|ds<\infty.
			\end{align}
			In particular, we have $|u(t)|<R$ for all $t\in I$, and, there is a $C>0$ such that $u(t)$ satisfies the energy estimate 
			\begin{align}
				|u(t)| + \gamma_{1}\int_{0}^{T_0}\frac{1}{s}|u(s)|ds\le \left( |u_0| + \int_{t}^{T_0}|\tilde{F}(s)|ds  \right)\text{e}^{C (T_0^p - t^p)/p},
			\end{align}
			for all $t\in I$. Moreover, given any $\mu\in \mathbb{R}$ with $\mu<p$ and $\mu\le\lambda$, the solution $u$ satisfies the decay estimate
			\begin{align}
				|t^{-\mu}u(t)|\le \text{e}^{C T^{p-\mu}_0/p} \int_{0}^{T_0}s^{-\lambda}|\tilde{F}(s)|ds
			\end{align}
			for all $t\in I$.
		\end{Thm} 
		The forwards Fuchsian theorem allows one to prove existence of solutions (to $F[\phi,t]=0$) with the particular asymptotic structure
		\begin{align}
			\phi(t)=\phi_\star(t)+f(t)u(t).
			\label{Eq:AsymptoticStruct}	
		\end{align}
		In particular, it ensures that $u(t)$ is small near $t=0$ and as such the function $\phi_{\star}(t)$ can be understood as describing the leading order behaviour of the unknown $\phi(t)$ near the singularity. However, it does not provide any information about the stability properties of the solutions. For this, one writes $\phi=\mathring{\phi}+\hat{\phi}$ where $\mathring{\phi}$ is the exact solution (whose existence was provided by the Forwards Fuchsian theorem) and $\hat{\phi}$ is some unknown function which is determined as a solution of the equation $F[\mathring{\phi}+\hat{\phi},t]$. Here, one aims to show that the resulting equation for $\hat{\phi}$ is again of Fuchsian form. If it is one can apply the following theorem:
		\begin{Thm}[Backwards Fuchsian Theorem]
			\label{Thm:BcksFuchs}
			Consider constants $T_{0}>0,p>0,$ and $R>0$, and a Fuchsian ODE system 
			\begin{align}
				\partial_{t}u(t)=\frac{1}{t}\mathcal{B}(t,u(t))u(t) + \tilde{F}(t) + t^{p-1}F_{0}(t,u(t)), 
				\quad
				u(T_0)=u_{0}
				\label{Eq:FuchsianIVP}
			\end{align}
			in the sense of \Defref{Def:FuchsianODE}. Suppose in addition that 
			\begin{align}
				\gamma_{1}|v|^2\le (v,\mathcal{B}(t,v) v),
				\label{Eq:FuchsianODE_Bck_y1}
			\end{align} 
			for all $v\in B_{R}(\mathbb{R}^N)$ and $t\in I$.
			
			Then there exists a $\delta>0$ such that the initial value problem \Eqref{Eq:FuchsianIVP} has a unique global continuously differentiable solution $u:I\rightarrow \mathbb{R}^N$ provided 
			\begin{align}
				|u_{0}|\le\delta,
				\quad
				\int_{0}^{T_0}|\tilde{F}(s)|ds\le\delta.
			\end{align} 
			In particular, we have $|u(t)|<R$ for all $t\in I$, and, there is a $C>0$ such that $u(t)$ satisfies the energy estimate 
			\begin{align}
				|u(t)| + \gamma_{1}\int_{0}^{T_0}\frac{1}{s}|u(s)|ds\le \left( |u_0| + \int_{t}^{T_0}|\tilde{F}(s)|ds  \right)\text{e}^{C (T_0^p - t^p)/p},
			\end{align}
			for all $t\in I$. Moreover, suppose that there is a constant $\lambda\ge0$ such that 
			\begin{align}
				\int_{0}^{T_0}s^{-\lambda}|\tilde{F}(s)|<\infty.
			\end{align}
			Then, given any $\mu\in \mathbb{R}$ with $\mu<p$ and $\mu\le\gamma_{1}$ and $\mu\le\lambda$, the solution satisfies 
			\begin{align}
				|u(t)|\le C t^{\mu}
			\end{align}
			for all $t\in I$.
		\end{Thm} 
		The theorem tells us that if the perturbed unknown $\hat{\phi}$ is sufficiently small at the initial time $t=T_0$ then it remains small as $t\rightarrow0$. In this sense Theorem~\ref{Thm:BcksFuchs} can be used to provide information about the stability properties of the solution $\mathring{\phi}$.
	\end{appendices}
\end{document}